\documentclass[11pt]{article}
\usepackage[margin=1in]{geometry}

\usepackage{float}
\usepackage{graphicx}
\usepackage{subcaption}
\usepackage{pgfplots}
\pgfplotsset{compat=1.18}
\usepackage{caption}
\usepackage{xcolor}
\usepackage{amsmath}
\usepackage{amssymb}
\usepackage{amsfonts}
\usepackage{dsfont}
\usepackage{cases}
\usepackage{amsthm}
\usepackage{upgreek}
\usepackage{url}    
\usepackage{placeins}

\usepackage{hyperref} % enables and customizes hyperlinks
\hypersetup{
	colorlinks   = true, 
	citecolor    = blue, 
	linkcolor    = red, 
	urlcolor     = blue
}
\usepackage{mathtools}                      % need for `show only references'
\mathtoolsset{showonlyrefs=true}

\usepackage[authoryear]{natbib}
\usepackage{comment}

\DeclareMathOperator*{\argmin}{arg\, min}

\newtheorem{theorem}{Theorem}[section]
\newtheorem*{theorem*}{Theorem}
\newtheorem{corollary}{Corollary}[section]
\newtheorem{lemma}{Lemma}[section]
\newtheorem{problem}{Problem}[section]
\newtheorem{proposition}{Proposition}[section]
\newtheorem{assumption}{Assumption}[section]
\newtheorem{definition}{Definition}[section]
\newtheorem{remark}{Remark}[section]
\newtheorem{example}{Example}[section]

\newcommand{\Eb}{\mathbb{E}}
\newcommand{\Pb}{\mathbb{P}}
\newcommand{\Rb}{\mathbb{R}}

\newcommand{\Bc}{\mathcal{B}}
\newcommand{\Cc}{\mathcal{C}}
\newcommand{\Fc}{\mathcal{F}}
\newcommand{\Ic}{\mathcal{I}}
\newcommand{\Nc}{\mathcal{N}}
\newcommand{\Tc}{\mathcal{T}}
\newcommand{\Vc}{\mathcal{V}}
\newcommand{\Xc}{\mathcal{X}}

\newcommand{\etat}{\tilde{\eta}}

\newcommand{\bc}{\mathbf{c}}
\newcommand{\ch}{\widehat{\bc}}
\newcommand{\rc}{\mathbf{\bar{c}}}

\newcommand{\dd}{\mathrm{d}}
\newcommand{\ee}{\mathrm{e}}
\newcommand{\mm}{\mathrm{m}}

\newcommand{\ndo}{n_{\downarrow}}
\newcommand{\nup}{n_{\uparrow}}

\begin{document}

\title{%Bonus–Malus Insurance in Oligopoly:
Optimal Underreporting and Competitive Equilibrium}

\author{Zongxia Liang\thanks{Department of Mathematical Sciences, and Center for Insurance and Risk Management, School of Economics and Management, Tsinghua University,  China. Email: liangzongxia@tsinghua.edu.cn} 
	\hspace{2ex}
	Jiayu Zhang\thanks{Department of Mathematical Sciences,  Tsinghua University,  China. Email: zjy23@mails.tsinghua.edu.cn}
	\hspace{2ex}
	Zhou Zhou\thanks{School of Mathematics and Statistics,  University of Sydney,  Australia. Email: zhou.zhou@sydney.edu.au} 
	\hspace{2ex}
	Bin Zou\thanks{Department of Mathematics,  University of Connecticut,  USA. Email: bin.zou@uconn.edu}}

\date{ \today}
\maketitle

\begin{abstract}
This paper develops a dynamic insurance market model comprising two competing insurance companies and a continuum of insureds,  and examines the interaction between strategic underreporting by the insureds and competitive pricing between the insurance companies under a Bonus-Malus System (BMS) framework. For the first time in an oligopolistic setting,  we establish the existence and uniqueness of the insureds' optimal reporting barrier,  as well as its continuous dependence on the BMS premiums. For the 2-class BMS case, we prove the existence of Nash equilibrium premium strategies and conduct an extensive sensitivity analysis on the impact of the model parameters on the equilibrium premiums. 
\end{abstract}


\section{Introduction}

Insurance companies commonly employ bonus–malus systems (BMS) that link future premiums to insureds’ historical claims in order to improve risk classification and to provide incentives for enhancing loss prevention,  thereby mitigating ex-ante moral hazard.  Such mechanisms allow premiums to more accurately reflect the true risk profile of the insureds.  For example,  \cite{abbring2003moral} document that claim frequency is negatively correlated with past claims and find no evidence of moral hazard, when analyzing French automobile insurance data under BMS.
The prototypical BMS operates as follows: if an insured does not file a claim in the preceding period,  the premium for the subsequent period is reduced (``bonus"); conversely,  if a claim occurs,  the premium is increased (``malus").  Within this framework,  insureds face strategic incentives to underreport losses.  Specifically,  an insured will choose not to report a loss if the expected increase in the subsequent premium caused by reporting exceeds the immediate benefit of the claim; on the other hand,  reporting is optimal if the expected premium increase is smaller than the claim benefit.  A substantial body of empirical and theoretical research has documented the existence of or studied such strategic underreporting by insureds (see, e.g., \cite{haehling1974optimal}, \cite{abbring2008better}, and \cite{robinson2010moral}, among many others). 

Building on the aforementioned phenomenon of strategic underreporting,  a natural question arises regarding the insureds’ optimal reporting decisions. Research on this topic remains relatively scarce, and we review recent progress as follows.
\cite{zacks2004claiming} employ standard dynamic programming methods to investigate the optimal reporting strategy and its impact on the insurer’s long-run average premium in multi-class BMS under risk neutrality with deductible insurance.  \cite{ludkovski2010ex} introduce a two-period model with asymmetric information,  in which the insured adopts a stochastic reporting strategy modeled as a Bernoulli random variable,  while the insurer updates beliefs about the risk type using Bayesian inference based on reported claims.  Their results show that the optimal reporting strategy for a risk-neutral insured varies across risk types,  and may involve full non-reporting,  full reporting,  or a mixed strategy.  \cite{charpentier2017optimal} provide a rigorous mathematical formulation of the claim reporting problem in discrete-time BMS settings.  Using numerical methods,  they compute the optimal reporting strategy and conduct sensitivity analyses for both the optimal strategy and underreporting probabilities.  \cite{cao2024equilibrium} derive expressions for the optimal reporting strategy for insureds holding full-coverage insurance,  providing a closed-form solution under risk neutrality and a semi-explicit form under risk aversion.  \cite{cao2024strategic} further investigate the optimal reporting strategy for insureds with deductible insurance,  presenting conditions for non-zero reporting and a semi-explicit expression,  and perform sensitivity analyses with respect to the optimal deductible. All above reviewed papers consider a discrete-time model; for optimal reporting problems in continuous-time, please refer to \cite{cao2025continuous} and \cite{cao2025optimal}. 

From the above summary, it is apparent that the existing literature on underreporting under the BMS frameworks focuses exclusively on insureds’ reporting strategies within a \emph{monopolistic} insurance market (i.e., there is only one insurance company). However, real insurance markets are typically oligopolistic,  characterized by intense competition among insurers.  This gives rise to two fundamental questions:
\begin{itemize}
    \item Question 1: How should insureds’ reporting strategies be characterized in an oligopolistic insurance market?
    \item Question 2: When insureds follow their optimal reporting strategies, how should competing insurance companies set the premiums for their individual BMS to reach (Nash) equilibrium?
\end{itemize}

To address these two questions above, we develop a discrete-time insurance market model with two competing insurance companies and a continuum of insureds. All insureds purchase full insurance but may choose either company as the provider. Both companies apply an $N$-class BMS in pricing; however, their premiums may differ over rate classes. Because factors, such as service quality and the complexity of claims procedures, generate heterogeneous preferences across insurance companies (see, e.g.,  \cite{ennew1996impact} and \cite{cummins1996capital}),  we incorporate a choice function \(\eta\) to capture such insurer-specific preferences (here, $\eta$ is a function, which takes the premium difference from the two insurers as its argument and returns a probability for choosing Company~1 over Company~2).  Following the approach of \cite{zacks2004claiming},  \cite{charpentier2017optimal}, and \cite{cao2024strategic},  we assume that insureds employ a barrier reporting strategy when making reporting decisions.  Let \(b=\{b_n^i\}_{i=1, 2, \, n=1, \cdots, N}\) denote a barrier reporting strategy, with \(b_n^i \ge 0\) being the barrier when the insured is in rate class \(n\) and buys insurance from Company \(i\). Then, such a strategy dictates the insured in  rate class \(n\) with Company \(i\) to report a loss if and only if it is greater than the barrier \(b_n^i\). 

For Question 1, we seek an optimal barrier strategy to minimize the expected total discounted expenses (i.e., insurance premiums plus hidden losses) and obtain a complete answer in Theorem~\ref{thm:b_op}, with additional results collected in Remarks~\ref{rem:b} and \ref{rem:obj}.
We show that the optimal reporting barrier \(b^*\) exists and is unique.  In addition, we derive the condition under which \(b^*\) is strictly positive,  along with its semi-explicit expression.  
The characterization of $b^*$ has a clear economic meaning: it is the difference in expected expenses (or gain in expected utility) between reporting a loss and hiding a loss.

Regarding Question 2, which, to the best of our knowledge,  has not been studied before in the  literature, we consider a 2-class BMS (i.e., $N=2$) and assume an exponential choice function $\eta$ for tractability.  In this setting,  we first obtain a finer result on the insureds' optimal reporting barrier in Theorem~\ref{thm:bstar}; next, we apply the Stackelberg game framework to study the interaction between the insurance companies and insureds. To be precise, the two insurance companies are the leaders and set the BMS premiums $\bc^i = \{c^i_n\}_{n=1, \cdots, N}$, for $i=1,2$; the insureds are the followers and choose their optimal reporting barrier $b^*(\bc^1, \bc^2)$. Both companies aim to maximize their expected (per-period) profit, taking into account the possible underreporting from the insureds. Because of the competition in the market, each insurance company's optimal premium strategy from the Stackelberg game depends on its competitor's strategy, taking the form of $\rc^1(\bc^2)$ and $\rc^2(\bc^1)$, respectively. Finally, a Nash equilibrium premium strategy $(\bc^{*,1}, \bc^{*,2})$ is defined as a fixed point of the mapping $(\bc^1, \bc^2) \mapsto (\rc^1(\bc^2), \rc^2(\bc^1))$. 
We show that such an equilibrium premium strategy exists in Theorem~\ref{thm:c_op}, when the model parameters satisfy certain regularity conditions. We provide a concrete example to demonstrate that those regularity conditions can hold under a reasonable market. 
Furthermore,  we conduct an extensive sensitivity analysis to examine how various model parameters affect the insurers’ equilibrium premiums.

%when the insured’s objective is to minimize the discounted sum of losses,  we obtain a closed-form expression for the optimal reporting barrier (see Theorem~\ref{thm:bstar}).  Because the insured’s current state depends only on whether a claim was reported in the previous period,  the resulting state process is a finite-state,  recurrent Markov chain,  from which we derive the stationary distribution \(p\) of insured's states.  Building on this foundation,  and under suitable regularity conditions,  we establish the existence of an equilibrium premium.  

The remainder of the paper proceeds as follows.  Section \ref{sec:model} introduces the model framework.  Section  \ref{sec:insured} derives the optimal reporting strategy for the insureds.  Section \ref{sec:insu} studies the insurance companies' pricing game.  Section \ref{nubs} presents numerical analyses and sensitivity results on the equilibrium premium strategies.  Section \ref{conl} concludes the paper. All proofs are placed in Appendix \ref{sec:proofs}.

\section{Model}
\label{sec:model}

We consider an insurance market consisting of two \emph{competing} insurance companies (insurers) and a continuum of insureds who seek insurance coverage on some non-life risk (such as automobile collision risk). On an infinite time horizon $T := \{1,  2,  \cdots\}$,  we model each insured's losses by a series of independently and identically distributed (i.i.d) nonnegative random variables,  $\{L_t\}_{t \in \Tc}$,  in which $L_t \overset{d}{=}L $ denotes the loss amount in the $t$-th period (over $[t-1,  t)$) and has the same distribution as a generic random variable $L$. We assume that $L$ is a mixture of a point mass at zero,  with probability $p_0 \in (0, 1)$,  and a continuously distributed positive random variable with full support over $(0, \infty)$ (see, e.g.,  \cite{haehling1974optimal} and \cite{cao2024equilibrium}). With this assumption,  the cumulative distribution function (cdf) $F_L$ of $L$ is given by 
\begin{align}
 \label{F_L}
	F_L(x) = p_0 + \int_0^x  f_L(\ell)  \dd \ell, \quad x \ge 0,
\end{align}
in which $f_L(\ell) > 0$ for all $\ell >0$ and $\lim_{\ell \rightarrow +\infty} f_L(\ell) = 0$. 
We fix a filtered probability space  $(\Omega, \Fc,  \mathbb{F},  \Pb)$ consistent with the above loss model, where $\mathbb{F}=\{\Fc_t\}_{t \in T}$ is a filtration; denote expectation taken under $\Pb$ by $\Eb$,  and any subscript of $\Eb$ indicates a conditional expectation.

The two insurance companies offer \emph{full insurance} covering the loss $L$ for the insureds,  and they each apply an $N$-class bonus-malus system (BMS) to price their policies (see, e.g.,  \cite{lemaire2012bonus} for a standard reference on BMS). Denoting the set of rate classes and the index set of companies by 
\begin{align*}
	\Nc := \{1,  2,  \cdots,  N\} \quad \text{and} \quad \Ic := \{1,  2\}, 
\end{align*}
respectively,  Company $i \in \Ic$ sets the premium,  $c^i_n > 0$, for all insureds in rate class $n \in \Nc$ over each single period. Without loss of generality,  assume that class 1 is the best rating class,  while class $N$ is the worst one; as such,  the premiums should satisfy 
\begin{align}
	\label{eq:cond_c}
	0 < c_1^i \le c_2^i \le \cdots \le c_N^i,  \quad \text{for } i \in \Ic. 
\end{align}
For every insured,  given their current rate class $n \in \Nc$,  the BMS sets their rate class in the next period by the following rule: 
\begin{align}
	\begin{cases}
		\max\{n-1, 1\}, &\text{if there is no \emph{reported} loss in the current period};\\
		\min\{n+1, N\}, &\text{otherwise}.
	\end{cases}
\end{align}
Note that the update of rate class relies on the \emph{reported},  not \emph{incurred},  loss in each period; this setting is consistent with the practice in most non-life insurance lines,  such as automobile insurance and home insurance,  because insurers cannot monitor the actual loss status of the insureds (or it is too costly for them to implement monitoring or audit records). 

Because premiums are cheaper for insureds in good rate classes,  the BMS mechanism may incentivize some insureds to deliberately hide incurred losses,  so that they remain in or get updated to good rate classes. The behavior of \emph{underreporting losses} is well documented in the insurance (see, e.g., \cite{cohen2005asymmetric} and \cite{braun2006modeling}),  and related theoretical studies show that the optimal reporting decision is to employ a \emph{barrier} strategy (see, e.g.,  Remark 2.1 in \cite{cao2025optimal}). Following this strand of literature,  we assume that insureds adopt a barrier strategy to decide whether they should report an incurred loss,  which we explain in detail as follows. Let $X = \{X_t\}_{t \in T}$ denote the state process of a representative insured,  where 
\begin{align}
	X_t \in \Xc := \{(n,  i) | n \in \Nc,  \,  i \in \Ic \}
\end{align}
records the insured's rate class $n$ and their insurance provider,  Company $i$,  in the $t$-th period; note that $X_t \in \Fc_{t-1}$ is known at the beginning of the $t$-th period,  for all $t \in T$. A barrier reporting strategy $b$ is a (time-homogeneous) function,  mapping every $x \in \Xc$ into a nonnegative number $b(x) \in \Rb_+$,  and dictates the insured to report a loss if and only if $L_t > b(X_t)$ in the $t$-th period and to hide a loss otherwise.

Recall that there are two competing insurance companies in the market,  both offering the same insurance coverage,  but possibly at different premiums. In theory,  insureds should simply go with the cheaper supplier; however,  empirical evidence suggests that practical factors,  such as service quality,  company reputation,  and loyalty,  may deter insureds from selecting the insurer with the lowest premium (see, e.g., \cite{cummins1996capital} and \cite{ennew1996impact}). With this in mind,  we propose a random model for the switching of insureds from one company to the other,  in which the transition probability depends on the difference of the premiums charged by the two companies and is independent of the loss. To be precise,  let $\eta : \Rb \mapsto [0, 1]$ be a continuous, nondecreasing and almost everywhere differentiable function; the probability that an insured in rate class $n$ chooses Company $1$   is given by $\eta(c_n^2 - c_n^1)$,  for all $n \in \Nc$. The nondecreasingness of $\eta$ captures the fact that the bigger the premium gap,  the more likely that insureds will switch to the cheaper company. One may impose further conditions on $\eta$,  such as $\lim\limits_{\Delta c \to +\infty} \eta(\Delta c)= 1$ and $\lim\limits_{\Delta c \to -\infty} \eta(\Delta c)= 0$. To fully model the state transition of insureds, we introduce the following function $\etat$:
\begin{equation}
\label{eq:eta}
\etat(c_n^i-c_n^j)=
\begin{cases}
1-\eta(c_n^2-c_n^1),  & i=1, \\
\eta(c_n^2-c_n^1),  & i=2, 
\end{cases}
\end{equation}
for all $n\in \Nc$ and $i, j \in \Ic$ with $i \ne j$.

Based on the above setup,  the insured's state $X$ follows a Markov chain,  with the transition probabilities given by
\begin{equation}
	\label{eq:P}
	\begin{split}
\Pb \left( X_{t+1}=((n-1)\vee 1, i)\, |\, X_{t}=(n, i) \right) &= \Pb (L_t\leq b(X_t)) \cdot \left(1-\etat \Big(c_{(n-1)\vee 1 }^i-c_{(n-1)\vee1}^j \Big) \right),  \\
\Pb \left( X_{t+1}=((n-1)\vee 1,  j)\, |\, X_{t}=(n, i) \right) &= \Pb (L_t\leq b(X_t)) \cdot \etat \Big(c_{(n-1)\vee 1}^i-c_{(n-1)\vee1}^j \Big),  \\
\Pb \left( X_{t+1}= ((n+1)\wedge N, i) \, |\, X_{t}=(n, i) \right) &= \Pb (L_t > b(X_t)) \cdot  \left(1-\etat \Big(c_{(n-1)\vee 1}^i-c_{(n-1)\vee 1}^j \Big) \right),  \\
\Pb \left( X_{t+1}= ((n+1)\wedge N,  j ) \, |\, X_{t}=(n, i) \right) &= \Pb (L_t > b(X_t)) \cdot  \etat \Big(c_{(n-1)\vee 1}^i-c_{(n-1)\vee 1}^j \Big) ,  
\end{split}
\end{equation}
for all $t \in T$,  $n \in \Nc$,  and $i \neq j \in \Ic$,  in which $m_1 \vee m_2 := \max\{m_1,  m_2\}$ and $m_1 \wedge m_2 := \min \{m_1,  m_2\}$ for all $m_1,  m_2 \in \Rb$. To reduce notational burden,  we introduce the following short-handed notation:
\begin{align}
	\label{eq:n}
	\ndo := (n-1)\vee 1 \quad \text{and} \quad \nup := (n+1)\wedge N,  \quad n \in \Nc, 
\end{align}
where the direction of the arrow indicates when the rate class decreases (better rating) or increases (worse rating),  and 
\begin{align}
	\label{eq:dc}
	\Delta c_n^{ij} := c_n^i - c_n^j,  \quad n \in \Nc,  i,  j \in \Ic,  i \neq j, 
\end{align}
denotes the premium difference on rate class $n$ between Company $i$ and Company $j$.

To study the decision-making of both the insureds and insurance companies,  we follow the Stackelberg game framework to account for both parties' interest; see, e.g., \cite{chen2018new},  \cite{li2021bowley},  \cite{cao2022stackelberg},  and \cite{boonen2023bowley},  among many others on the application of this game model in actuarial science. In our game model,  we assume that the insureds are the followers and choose their barrier reporting strategy $b$,  and that the insurance companies are the leaders and set the premiums $\bc^i = \{c^i_n\}_{n \in \Nc}$,  $i \in \Ic$. For every premium pair $(\bc^1,  \bc^2)$ set by the leaders (insurance companies),  the insureds seek an optimal barrier strategy $b^*:= b^*(\bc^1,  \bc^2)$,  which depends on $(\bc^1,  \bc^2)$,  to minimize the discounted total cost (premiums plus hidden losses); we analyze the insureds' problem in Section \ref{sec:insured}. Knowing the response $b^*(\bc^1,  \bc^2)$ in loss reporting to their premium strategies,  the insurance companies aim to maximize their expected profit,  and their competition is settled via a Nash game,  yielding the equilibrium premiums $(\bc^{*,1},  \bc^{*,2})$,  which in turn leads to the equilibrium reporting strategy $b^*(\bc^{*,1},  \bc^{*,2})$. The insurance companies' equilibrium premium problems are solved in Section \ref{sec:insu}.

 \section{The insureds' optimal reporting problem}
 \label{sec:insured}
 
In this section,  we formally introduce,  and then solve,  the insureds' optimal reporting problem,  when the premiums $(\bc^1,  \bc^2)$ are given from the insurance companies. The optimal barrier strategy $b^*:= b^*(\bc^1,  \bc^2)$ is obtained in Theorem \ref{thm:b_op}. 

In the analysis,  we consider a representative insured whose initial state $X_0$,  though fixed,  can take any value in the feasible set $\Xc$. Because the state space $\Xc$ is finite,  a barrier strategy $b$ (as a mapping) is equivalent to an $N \times 2$ matrix $\{b_n^i\}_{n \in \Nc,  i \in \Ic}$,  with the one-to-one correspondence $b(x = (n, i)) = b_n^i \ge 0$. For convenience,  we write $b = \{b_n^i\}_{n \in \Nc,  i \in \Ic}$ as a barrier strategy hereafter,  and we denote by $\Bc$ the set of all admissible barrier strategies, that is,
\begin{align}
    \label{B}
    \Bc = \bigl\{
(a_{ni}) \in \mathbb{R}_{+}^{N\times 2}
\;\big|\;
a_{n i} \ge 0, 
\ \forall\,  n\in \Nc, i \in \Ic
\bigr\}.
\end{align}
Let $\Cc$ denote the set of feasible premiums (formally defined by \eqref{eq:C} later). We formulate the insured's optimal reporting problem below. 

\begin{problem}
	\label{prob:cus}
	Given the premium pair $(\bc^1,  \bc^2) \in \Cc^2$ set by the two insurance companies,  the representative insured seeks an optimal barrier strategy $b^* := b^*(\bc^1,  \bc^2)$ that  minimizes the discount total cost 
	\begin{align}
		\label{eq:b_op_def}
		b^* = \argmin_{b \in \Bc} \,  \Eb \left[\sum_{t=1}^\infty \delta^t \left(c(X_{t}) + L_{t} \cdot \mathbf{1}_{\{L_{t} \leq b(X_t)\}} \right) \Big| X_1 \in \Xc \right] := \argmin_{b\in\Bc}V(X_1; b), 
	\end{align}
where  $\delta \in (0,1)$ is a discounting factor,  $\mathbf{1}_D$ is the indicator function of a set $D$, 
\begin{align}
	c(X_t) = c_{n}^{i} \quad \text{and} \quad b(X_t) = b_{n}^{i},  \quad \text{ for all } X_t = (n,  i) \in \Nc \times \Ic.
\end{align}
We call  $V^*(n,  i) := \inf_{b\in\Bc}V(n,  i; b)$  the \emph{value function} for the initial state $X_1 = (n,  i) \in \Xc$.
\end{problem}

We briefly explain the insured's problem as follows. At time $t - 1$,  the insured's state $X_t = (n,  i) \in \Nc \times \Ic$ is observable; that is,  they are in rate class $n \in \Nc$ and buy insurance for the $t$-th period from Company $i \in \Ic$,  with both known at the beginning of the $t$-th period. The total cost in the $t$-th period is the sum of the insurance premium $c(X_t)$ and the hidden loss $L_t$,  should it fall below the reporting barrier $b(X_t)$. 
%In practice,  premiums are paid at the beginning of each period,  while given the discrete-time model,  losses are settled at the end of each period; however,  for convenience reason,  we directly add the premium and hidden loss together. Indeed,  one may treat $c$ as the ``nominal'' premium and $\delta \cdot c$ as ``real'' premium. 
In regardless of their initial state $X_1$,  the insured may switch to a different company at any time and reach any rate class over the infinite horizon;  as such, their reporting decision must account for the entire state space $\Xc$,  not just for the initial state. This explains why every reporting strategy $b$ is a full matrix $\{b_n^i\}$ covering all possible $(n,  i) \in \Xc$. As the notation suggests,  we expect the insured's optimal strategy $b^*$ to depend on the premium pair  $(\bc^1,  \bc^2)$,  and if such dependence relation needs to be emphasized,  we write it as $b^*(\bc^1,  \bc^2)$.

We show that Problem \ref{prob:cus} admits a unique solution $b^*$,  and it satisfies certain continuity condition.

 \begin{theorem}\label{thm:b_op}
 For all $\bc:= (\bc^1,  \bc^2) \in \Cc^2$,  there exists a unique optimal barrier strategy $b^* = \{ b^{*, i}_n(\bc)\}_{n \in \Nc,  i \in \Ic}$ to the insured's reporting problem in \eqref{eq:b_op_def},  and the mapping,  ${\bf c}\mapsto b^{*, i}_n(\bc)$,  is continuous. 
 Moreover,  we have the following characterization of the optimal barrier strategy:
 \begin{align}
 	b^{*, i}_n &= 0 \vee \delta\Bigl[  \left( 1 - \etat \big( \Delta c^{ij}_{\nup}\big)\right) V^*(\nup,  i) + \etat \big( \Delta c^{ij}_{\nup}\big) V^*(\nup,  j)  \notag \\
    &\quad -\left( 1 - \etat \big( \Delta c^{ij}_{\ndo}\big)\right) V^*(\ndo,  i) 
 	- \etat \big( \Delta c^{ij}_{\ndo}\big) V^*(\ndo,  j)  \Bigr] := 0 \vee \varphi (n,  i), 
 	\label{eq:varphi}
 \end{align}
in which $j = 3 - i \in \Ic$,  $\ndo$ and $\nup$ (rating up or down by one class) are defined in  \eqref{eq:n},  and  $\Delta c_n^{ij}$ (premium difference) is given by \eqref{eq:dc}. 
 \end{theorem}

\begin{remark}
	\label{rem:b}
By Theorem \ref{thm:b_op}, $b_n^{*,i} > 0$ if and only if $\varphi(n, i)>0$, and this result is intuitive because $\varphi$ in \eqref{eq:varphi} can be interpreted as the ``difference in cost'' between hiding a loss and reporting it. By recalling the definition of $\etat$ in \eqref{eq:eta} and using \eqref{eq:varphi}, we get $\varphi(n, 1) = \varphi(n, 2)$ for all $n \in \Nc$. As such, the insureds' optimal barrier strategy is independent of the insurance companies from which they purchase insurance; that is, 
\begin{align}
	\label{eq:same_b}
	b_n^{*,1} = b_n^{*,2} \quad \text{for all } n \in \Nc.
\end{align}
Because of this symmetry result, we write $b^* = \{b^*_n\}_{n \in \Nc}$, without the company index $i \in \Ic$, as the insureds' optimal barrier strategy. 

In the special case of two rate classes ($N=2$), it is obvious from the definition of $\varphi(n,i)$ that $b_1^* = b_2^*$ (that is, the optimal barrier is the \emph{same} for both rate classes).
\end{remark}

 \begin{remark}
 	\label{rem:obj}
    The optimization objective presented in Problem~\ref{prob:cus} aims to minimize the insured's discounted total expenses, and it is widely used in the literature on loss (under)reporting (see, e.g., \cite{haehling1974optimal}, \cite{robinson2010moral}, and \cite{charpentier2017optimal}). The linearity in the objective (see \eqref{eq:b_op_def}) simplifies the optimization problem and helps obtain an analytical solution to the optimal reporting strategy. From the perspective of expected utility theory (EUT), this linear preference is equivalent to assuming that insureds are \emph{risk-neutral}; alternatively, one may apply a utility function to the insureds' wealth (expenses), as suggested by the standard EUT. \cite{cao2024equilibrium} consider both risk-neutral and risk-averse insureds in their study; see Sections 3 and 4 therein for comparison. 
    
    Suppose that insureds are risk-averse in the sense that they apply a (strictly increasing) utility function $U$ to the expenses in each period. Instead of \eqref{eq:b_op_def}, they now solve the following reporting problem:
\begin{align*}
		%\label{eq:b_op_def_U}
		\tilde{b}^* = \arg\min_{b \in \Bc} \,  \Eb \left[\sum_{t=1}^\infty \delta^t \, U\left(c(X_{t}) + L_{t} \cdot \mathbf{1}_{\{L_{t} \leq b(X_t)\}} \right) \Big| X_1 \in \Xc \right]:=\tilde{V}(X_1;b) .
	\end{align*}
It is pleasing to report that all the results in Theorem \ref{thm:b_op} hold in a parallel way. Indeed, the above $\tilde{b}^*$ is unique, and it is given by 
 \begin{equation}
 	\label{eq:b_op_U}
     \tilde{b}_n^{*,i} = 0\vee \left(U^{-1}\left[\tilde{\varphi}(n, i)+U(c_n^i)\right]-c_n^i\right),  \quad n\in \Nc, i \in \Ic,
 \end{equation}
where $\tilde{\varphi} (n,  i)$ is defined similarly to $\varphi(n,i)$ with $V^*$ replaced by the new value function $\tilde{V}^*$.
 \end{remark}

\section{The insurance companies' pricing game}
\label{sec:insu}

We obtain the insureds' optimal barrier strategy $b^*(\bc^1, \bc^2)$ for reporting losses in Theorem \ref{thm:b_op}, in response to the premium strategies $(\bc^1, \bc^2)$ from two competing insurance companies. Knowing this optimal response function $b^*(\bc^1, \bc^2)$, the two insurance companies set their own premiums for the BMS to maximize the expected profit from the insurance business. As described in the setup, this pricing game is modeled as a two-step Stackelberg-Nash game: in the first step, Company $i$ ($i=1,2$) solves its respective Stackelberg game with the insureds to obtain $\rc^i(\bc^{3-i})$, when its competitor adopts strategy $\bc^{3-i}$; in the second step, the two companies set the Nash equilibria $(\bc^{*,1}, \bc^{*,2})$ as their final premiums. The goal of this section is to study this pricing game.

We first define the set of admissible premiums for each insurance company as follows:
\begin{align}
	\label{eq:C}
	\Cc = \left\{ (c_1,  \cdots,  c_N) \in \Rb_+^N : c_1 \le c_2 \le \cdots \le c_N \text{ and } c_n \in [\Eb[L],  M] \text{ for all } n \in \Nc \right\}, 
\end{align}
in which $M > 0$ is a constant. In the above definition,  the first condition is exactly \eqref{eq:cond_c},  as required by the ranking of all rate classes; in the second condition,  $c_n \ge \Eb[L]$ is the so-called nonnegative loading condition (insurance companies ruin for sure if this fails),  while $M$ is an upper bound on the premium so that insureds prefer purchasing full insurance to self-insurance.

For every admissible premium pair ${\bc}=(\bc^1, \bc^2) \in \Cc^2$, let $b^*(\bc)$ denote the insureds' optimal barrier reporting strategy as characterized by \eqref{eq:varphi}. Recall from \eqref{eq:same_b} that $b^*(\bc)$ is the same for insureds of both insurance companies. We obtain Company $i$'s expected per-period profit from premium strategy $\bc^i = \{c_n^i\}_{n \in \Nc}$, when Company $j$ adopts $\bc^j$, by 
\begin{align}
	\label{eq:f_obj}
	J^i({\bf c}^i;{\bf c}^j):=\Eb\left[\sum_{n=1}^N \left(c_n^i-L\cdot \mathbf{1}_{\{L > b_n^*(\bc)\}}\right)p(n, i; \bc)\right], \quad i = 1, 2, \, j = 3-i,
\end{align}
   where $p(n,i;\bc)$ denotes the long-run proportion of insureds in rate class $n\in\Nc$ who are insured
with Company $i$, given the premium pair $\bc$. Under the optimal barrier reporting strategy $b^*$,
the insureds’ state process evolves as a time-homogeneous Markov chain on a finite state space. 
Standard Markov chain arguments then imply that the cross-sectional distribution of insureds converges to a unique stationary distribution. Let $\Tc_{b^*}$ denote the corresponding transition
matrix. The stationary distribution vector $p$ is characterized as the solution to
\[
(I - \Tc_{b^*}^\top)p = 0,
\]
together with the normalization condition $\mathbf{1}^\top p = 1$. Consequently, $p(n,i;\bc)$
represents the steady-state proportion of insureds in class $n$ with Company $i$, endogenously
induced by the optimal reporting behavior and the premium pair $\bc$.

\begin{definition}
\label{def_c}
For every fixed $\bc^2 \in \Cc$, denote $ \rc^1 (\bc^2) := \arg\sup\limits_{\bc^1 \in \Cc}\, J^1(\bc^1;\bc^2)$ as Company 1's optimal premium strategy; for every fixed $\bc^1 \in \Cc$, denote $ \rc^2 (\bc^1) := \arg\sup\limits_{\bc^2 \in \Cc}\, J^2(\bc^2;\bc^1)$ as Company 2's optimal premium strategy. A premium pair $(\bc^{*,1}, \bc^{*,2})$ is called an equilibrium pricing strategy if it is a fixed point of the mapping $(\bc^1, \bc^2) \mapsto (\rc^1(\bc^2), \rc^2(\bc^1))$.
\end{definition}

The equilibrium pricing strategy $(\bc^{*,1}, \bc^{*,2})$ is defined as an equilibrium to a (non-cooperative) Nash game, and such a definition is well recognized in the economics literature. However, finding an analytical solution for $(\bc^{*,1}, \bc^{*,2})$ under a general BMS is extremely challenging. Indeed, several recent works that study loss underreporting in discrete time only obtain analytical results when the BMS has two classes (see, e.g., \cite{cao2024equilibrium}). This motivates us to study the insurance companies' pricing game for a 2-class BMS. To improve readability, we place all the proofs in Appendix \ref{sec:proofs}.

In the rest of this section, we assume that the BMS has two rate classes (with $N = 2$), and Class 1 is the ``good'' class, while Class 2 is the ``bad'' class. We further assume that both companies apply the same proportional penalty to the bad class; that is, for all $\bc^1 = \{c_1^1, c_2^1\} \in \Cc$ and $\bc^2 = \{c_1^2, c_2^2\} \in \Cc$, there exists a common penalty $\kappa \in (1,2)$ such that 
\begin{align}\label{k}
	c_2^1 = \kappa \, c^1_1 \quad \text{and} \quad  c_2^2 = \kappa \, c_1^2 .
\end{align} 
(We impose an upper bound of 2 on the penalty $\kappa$ for practical reasons.)
For notational convenience, we set 
\begin{align*}
	\theta_1 := c^1_1 \quad \text{and} \quad \theta_2 := c_1^2,
\end{align*}
and with the above assumption, every $\bc^1$ (or $\bc^2$) is uniquely linked with a one-dimensional parameter $\theta_1$ (or $\theta_2$), which takes values from the following set (as a result of \eqref{eq:C}):
\begin{align}
 \label{Theta}
	\Theta := \{ \theta \in \Rb_+ : \theta \in [\Eb[L], \, M/ \kappa] \}.
\end{align}

In the model of Section \ref{sec:model}, the probability that an insured chooses Company 1 over Company 2 is captured by a general, nondecreasing function $\eta$, which yields the probability of switching companies in \eqref{eq:eta}. Here, we assume that $\eta$ is of the following exponential form:
\begin{align}\label{k1k2}
	\eta(\theta_2 - \theta_1) = \begin{cases}
		k_2 \, \ee^{k_1(\theta_2 - \theta_1)},  & \theta_1 > \theta_2,  \\
		1-(1-k_2) \, \ee^{-k_1(\theta_2 - \theta_1)},   & \theta_1 \le \theta_2,
	\end{cases}
\end{align}
for all $\theta_1, \theta_2 \in \Theta$, where  the  $\eta$, $k_2 \in (0, 1)$ is the probability that insureds choose Company 1 when two companies offer the same premiums, and $k_1 \in (0,1)$ represents the sensitivity degree of insureds to the premium difference on the choice of an insurance provider. Note that a larger value of $k_1$ implies that insureds are more likely to choose the cheaper provider.

We summarize the specifications of the above 2-class BMS as follows.
\begin{assumption}
	\label{asu:2-class}
	The BMS has two rate classes ($n = 1, 2$), and the premium strategies of two insurance companies, $\bc^1 = \{c_1^1, c_2^1\} \in \Cc$ and $\bc^2 = \{c_1^2, c_2^2\} \in \Cc$, take the form of  
	\begin{align}
		\label{eq:theta}
		c_1^1 = \theta_1, \, c_2^1 = \kappa \theta_1, \text{ and }
		c_1^2 = \theta_2, \, c_2^2 = \kappa \theta_2, 
	\end{align}
   in which $\kappa \in (1,2)$ and $\theta_1, \theta_2 \in \Theta$ (implying that $\bc^i$ is one-to-one with $\theta_i$, $i = 1, 2$). Given the premiums $\theta_1$ and $\theta_2$ on Class 1, the probability that insureds choose Company 1 over Company 2 is given by the $\eta$ function in \eqref{k1k2}.
\end{assumption}

There are four states for the insured's Markov state $X = (n, i)$, denoted by 
\begin{align}
	\label{eq:four_states}
	s_1 = (1,1), \; s_2 = (2,1), \; s_3 = (1, 2), \; s_4 = (2,2),
\end{align}
in which $n$ records their current rate class, and Company $i$ is their insurance provider. Given Assumption~\ref{asu:2-class}, we can strength the result of Theorem \ref{thm:b_op} and obtain the insured's optimal reporting strategy in closed form.

\begin{theorem}
	\label{thm:bstar}
	Let Assumption~\ref{asu:2-class} hold and $(\theta_1, \theta_2) = (c_1^1, c_1^2) \in \Theta^2$ be given. The unique optimal barrier reporting strategy $b^*(\theta_1, \theta_2)$ for Problem~\ref{prob:cus} is the same across all four states $s_j$ ($j=1,2,3,4$), and it is given by 
	\begin{align}
		b^*(\theta_1, \theta_2) = \delta  (\kappa - 1) \left[ \theta_1 \cdot \eta(\theta_1 - \theta_2) + \theta_2 \cdot \big( 1 - \eta(\theta_1 - \theta_2) \big)\right] > 0 .
	\end{align}
\end{theorem}

Now knowing exactly how insureds underreport losses by the barrier strategy $b^*(\theta_1, \theta_2)$, the two insurance companies optimize their objective $J^i$ in \eqref{eq:f_obj}. We solve $\sup\limits_{\bc^1 \in \Cc} \, J^1(\bc^1; \bc^2)$ and  $\sup\limits_{\bc^2 \in \Cc} \, J^2(\bc^2; \bc^1)$ and identify the conditions under which the optimizers $\rc^1(\bc^2)$ and $\rc^2(\bc^1)$ exist. Then,  combining these conditions, we obtain the existence of the equilibrium premium strategy $(\bc^{*,1}, \bc^{*,2})$, as presented in the next theorem. Denote 
\begin{equation}
	\label{eq:m}
	\begin{split}
	\mm_1 &:= \frac{k_1}{\delta (\kappa-1)^2 k_2 [(2 A_1) \vee (2 \delta -A_1)]}, \, \hskip 1cm  A_1 := \frac{k_2 (2- \delta k_2)}{2-k_2}
     , \\
	\mm_2 &:= 
	\frac{k_1}{\delta (\kappa-1)^2 (1-k_2) [(2 A_2) \vee (2 \delta -A_2)]}, \, A_2 := \frac{(1-k_2) (2- \delta (1-k_2))}{1+k_2}.
	\end{split}
\end{equation} 

\begin{theorem}
	\label{thm:c_op}
	Let Assumption~\ref{asu:2-class} hold. There exists an equilibrium premium strategy $(\bc^{*,1}, \bc^{*,2})$ defined in Definition~\ref{def_c} if the following three conditions hold: 
	\begin{itemize}
		\item[(i)] \quad
		$\displaystyle 
		\sup_{ \ell \in I_{L}}
		%{\delta (\kappa - 1) \Eb[L] \le \ell \le \delta \frac{\kappa - 1}{\kappa} M }
		\, 
		\frac{f_L'(\ell)}{f_L(\ell)}
		\in \Big[-\frac{\ee k_1}{2\ee-1} , \, \frac{k_1}{2}\Big]$, with $I_L := \left[\delta   (\kappa-1)  \Eb[L], \,  \frac{\delta (\kappa-1) M}{\kappa} \right]$;
		
		\item[(ii)] \quad
		$\displaystyle 
			\sup_{\ell \in I_{L}}
			%{\delta (\kappa - 1) \Eb[L] \le \ell \le \delta \frac{\kappa - 1}{\kappa} M } \,
			 f_L(\ell) \le \mm_1 \wedge \mm_2$;
		
		\item[(iii)] \quad
		$\displaystyle 
		\frac{(1-\delta) k_1 M}{\kappa} \le A_1 \wedge A_2$,
	\end{itemize}
in which $\delta$ is the discount factor, $\kappa$ is the premium penalty in \eqref{k}, $M$ is the maximum on the premiums, $f_L$ is the density of the loss random variable $L$, $k_1 $and $k_2$ are from \eqref{k1k2}, and $\mm_1$, $\mm_2$, $A_1$ and $A_2$ are defined in \eqref{eq:m}.
\end{theorem}

\begin{remark}
	\label{rem:cond}
Condition (i) in Theorem~\ref{thm:c_op} indicates that when $\ell \in I_L$,  the derivative of the log density function is bounded over an interval multiplicative of parameter $k_1$; this condition requires the density $f_L$ to vary relatively smoothly when losses are in $I_L$. Losses from low-probability, high-severity insurance policies often follow heavy-tailed distributions,  which typically satisfy this slowly varying condition (an example is Gamma distribution, frequently used in modeling insurance losses).

For loss distributions with $f_L$ bounded above by $1-p_0$ (such as Pareto distributions with shape parameter less than or equal to the minimum loss),  Condition (ii) naturally holds, if the probability of losses, \(1-p_0\), satisfies \[1-p_0 \le  \frac{k_1}{2 \delta (2-\delta)[k_2 \vee (1-k_2)] (\kappa-1)^2}.\]
This condition is rather weak because $\kappa \in (1.2, 1.5)$ in most BMS models (see Chapter 12 of \cite{MeyersJones2015}). To further illustrate the reasonableness of the technical conditions in Theorem~\ref{thm:c_op}, we provide in Appendix~\ref{sec:example} a concrete example in which all these conditions are simultaneously satisfied under economically plausible parameter specifications.
\end{remark}

As seen from \eqref{eq:theta}, every premium strategy $(\bc^1, \bc^2)$ is uniquely linked with a pair $(\theta_1, \theta_2)$, with $\theta_i$ being the premium charged on Class 1 from Company $i$, $i=1,2$. For this reason, we call $(\theta_1^*, \theta_2^*)$ an equilibrium if it yields $(\bc^{*,1}, \bc^{*,2})$ in Definition~\ref{def_c}. We next show an interesting result on the relative ordering of the equilibrium premiums $\theta_1^*$ and $\theta_2^*$. Recall from \eqref{k1k2} that $k_2$ is the  probability that insureds choose Company 1 when both companies offer the same premium ($\theta_1 = \theta_2$). 

\begin{proposition}
    \label{p4.1}
    Suppose that an equilibrium premium strategy $(\theta_1^*, \theta_2^*)$ exists. If the density function $f_L (\ell)$ at strictly positive losses $\ell >0$ satisfies
    \[
        \sup_{\ell \in I_L}
        %{\delta (\kappa-1) \Eb[L] \le \ell \le \delta (\kappa-1) \frac{M}{k}}  
        \, \ell f_L (\ell)
        \le \frac{1}{(1-\delta)(\kappa-1)}, \, \, \,  with \, \, I_L := \left[\delta   (\kappa-1)  \Eb[L], \,  \frac{\delta (\kappa-1) M}{\kappa} \right],
    \]
    then the ordering of the equilibrium premiums is determined by \(k_2\) as follows: 
    $\theta_1^* \ge \theta_2^*$ when $k_2 \ge \frac{1}{2}$, and $\theta_1^* \le \theta_2^*$ when $k_2 \le \frac{1}{2}$. In particular, \(\theta_1^* = \theta_2^*\) when \(k_2 = \frac{1}{2}\).
\end{proposition}

\begin{remark}
	The ordering results are consistent with intuition and the definition of $k_2$. When $k_2 \ge 1/2$, Company 1 has the preference advantage over Company 2 among insureds, and this advantage allows Company 1 to set a higher premium than Company 2 in the competition. We also note that the condition on $f_L$ in Proposition~\ref{p4.1} holds under the conditions of Theorem~\ref{thm:c_op}, as shown by the inequalities below
	\begin{align*}
		\sup_{\ell \in I_L} \, \ell f_L (\ell) \le \frac{\delta (\kappa-1) M}{\kappa} \, \sup_{\ell \in I_L} \, f_L(\ell) \le \frac{1}{2 [k_2 \vee (1-k_2)] (1-\delta) (\kappa-1)} \le \frac{1}{(1-\delta)(\kappa-1)}.
	\end{align*}
\end{remark}

Theorem~\ref{thm:c_op} establishes the existence of an equilibrium premium strategy $(\bc^{*,1}, \bc^{*,2})$ (or equivalently, $(\theta_1^*, \theta_2^*)$), under three technical conditions. Remark~\ref{rem:cond} offers some explanations to these conditions; however, it remains unclear whether all three conditions hold at the same time, as required by Theorem~\ref{thm:c_op}, under practical specifications. In the rest of this section, we provide an example offering a positive answer to this question.

\section{Numerical analysis} \label{nubs}

Under the 2-class BMS in Section \ref{sec:insu},  Theorem \ref{thm:c_op} shows that the equilibrium premium strategy $(\bc^{*,1}, \bc^{*,2})$ in Definition~\ref{def_c} exists, when certain technical conditions hold. However, the equilibrium premiums depend strongly on the loss distribution, and they do not admit an analytical solution. As such, we conduct a numerical analysis in this section to gain further insights of the equilibrium premiums. Note that the analysis below focuses on the same 2-class BMS as in Section \ref{sec:insu}, under which a premium strategy $(\bc^1, \bc^2)=(\{c_1^1, c_2^1\}, \{c_1^2, c_2^2\})$ is one-to-one with $(\theta_1, \theta_2) \in \Theta^2$, as shown in \eqref{eq:theta}. For this reason, we call $(\theta_1^*, \theta_2^*)$ an equilibrium if it yields $(\bc^{*,1}, \bc^{*,2})$ in Definition~\ref{def_c}. 

To start,  we assume that the insured’s per-period  loss \(L\) is a mixture of a point mass at zero and a $Gamma(\alpha, \lambda)$ distribution (see Example~\ref{exm:condition}), with weights \(p_0 \in (0, 1)\) and \(1-p_0\), respectively.
We set the parameter values for the model as listed in Table~\ref{tab:parameters}, which constitutes the base case for our numerical study. When we study the impact of a particular parameter on the equilibrium $(\theta_1^*, \theta_2^*)$, we let this parameter vary over a reasonable range (around its base value) but keep all other parameters fixed as in Table~\ref{tab:parameters}. Recall that $k_1$ and $k_2$ are from the function $\eta$ in \eqref{k1k2}, and $\kappa$ is the premium penalty on Class 2 (see \eqref{eq:theta}).

\begin{table}[htb]
	\centering
	\begin{tabular}{lcc}
		\hline
		\textbf{Parameter} & \textbf{Symbol} & \textbf{Value} \\
		\hline
		\(\Pb (L=0)\) & $p_0$ & \(0.9\) \\
		$L|L>0 \sim \, Gamma(\alpha, \lambda)$ & \( (\alpha, \lambda ) \) & \((1.2, 0.0085)\) \\
		Upper bound on the premium & $M$ & \(35.853\) \\
		Price sensitivity & $k_1$ & \(0.015\) \\
		Preference for Company 1 & $k_2$ & \(0.8\) \\
		Premium penalty & $\kappa$ & \(1.25\) \\
		Discount factor & $\delta$ & \(0.97\) \\
		\hline
	\end{tabular}
	\caption{Model parameters in the base case}
	\label{tab:parameters}
\end{table}

Assuming that the model parameters are given in Table~\ref{tab:parameters}, we numerically compute the equilibrium premiums and obtain 
\begin{align*}
	\theta_1^* \approx 35.8293 \quad \text{and} \quad \theta_2^* \approx 33.4501.
\end{align*}
Note that we have $k_2 > 0.5$, and the above results confirm that $\theta_1^* > \theta_2^*$ as predicted by Proposition~\ref{p4.1}. Before we proceed to conducting sensitivity analysis, we note that the kinks (i.e., non-differentiable points) observed in the subsequent figures arise from the presence of an upper bound on premiums ($M$ in \eqref{eq:C}),  which truncates the curves that would otherwise vary smoothly with respect to the parameters.

\begin{figure}[htbp]
    \centering
    \begin{subfigure}[b]{0.45\textwidth}
        \centering
        \includegraphics[width=\textwidth, trim = 1cm 0cm 1cm 0.5cm, clip=true]{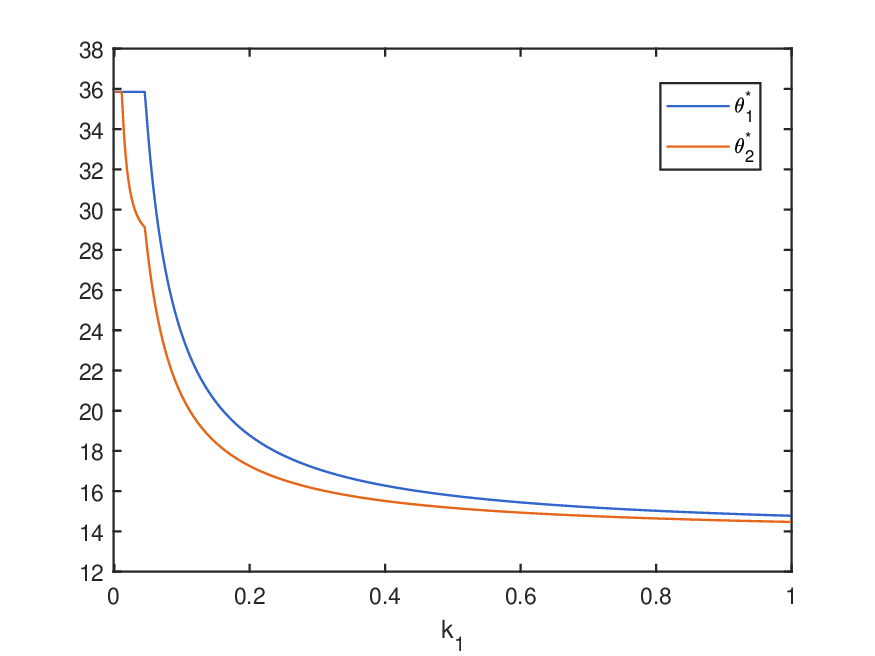}
    \end{subfigure}
    \hfill
    \begin{subfigure}[b]{0.45\textwidth}
        \centering
        \includegraphics[width=\textwidth, trim = 1cm 0cm 1cm 0.5cm, clip=true]{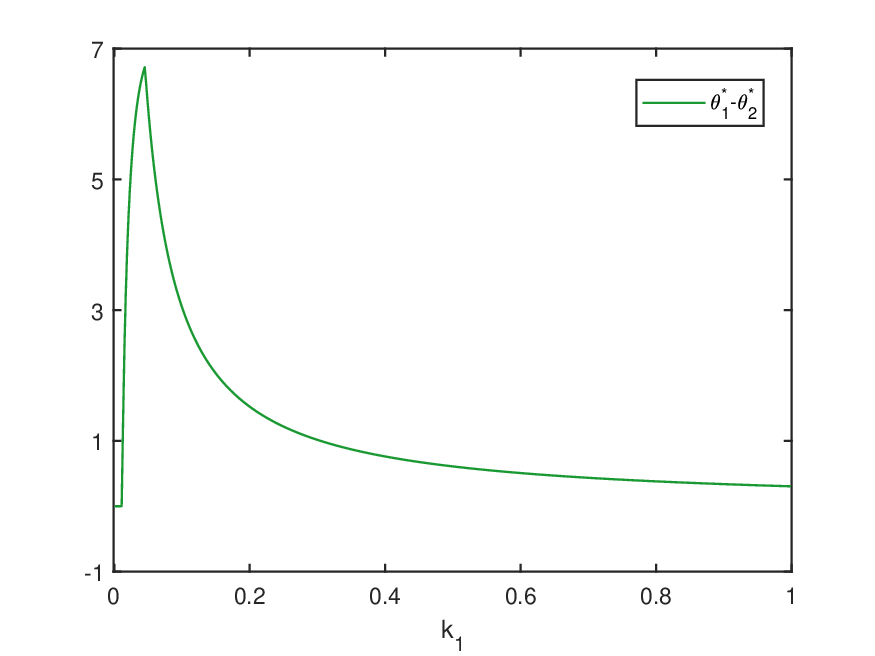}
    \end{subfigure}
\vspace{-2ex}
    \caption{Effect of the price sensitivity coefficient \(k_1\) on the equilibrium premiums}
    \label{fig:both1}
\end{figure}
In the first sensitivity study, we investigate how the price sensitivity coefficient \(k_1\) of the insureds affects the equilibrium premiums $\theta_1^*$ and $\theta_2^*$, and their difference. 
By \eqref{k1k2}, the bigger the coefficient $k_1$, the more likely that an insured chooses the cheaper insurance company to purchase insurance. The left panel of Figure~\ref{fig:both1} shows that the equilibrium premiums $\theta_1^*$ and $\theta_2^*$ remain constant initially and then decrease as $k_1$ increases. For very small $k_1$, the maximum premium constraint is binding, which explains the overlapping flat part. 
As $k_1$ increases,  however, insureds are more sensitive to the price difference between the two companies, and this naturally drives both companies to lower their premiums to attract insureds.
The right panel of Figure~\ref{fig:both1} shows that the premium difference $\theta_1^* - \theta_2^*$ is initially equal to zero and then exhibits a sharp rise, followed by a monotonic decline. 
When price sensitivity increases from very low levels, the market transitions from a regime dominated by intrinsic preference (or brand bias) to one in which price becomes a more important factor in the insured's choice on insurance provider.
In this transition region,  the disadvantaged insurer (Company~2) faces a higher effective demand elasticity, while the advantaged insurer (Company~1) benefits from a larger and more loyal base of customers. Consequently, Company~2 has a stronger incentive to reduce its premiums than Company~1 in the competition. This asymmetric adjustment leads to an increase in the equilibrium premium difference in the early stage of the transition. Such a behavior is consistent with the ``fat-cat'' strategy described by \cite{fudenberg1984fat},  whereby an incumbent with market power responds less aggressively to emerging competitive threats,  relying instead on its established advantage. 
However,  once $k_1$ exceeds a critical threshold, the price itself becomes the dominating factor when insureds choose their insurance provider, and the preference advantage of Company~1 over Company~2 weakens in the insureds' decision. In this high-sensitivity regime, both companies apply similar pricing strategies,  resulting in a rapid convergence of premiums.

\begin{figure}[htbp]
    \centering
    \begin{subfigure}[b]{0.45\textwidth}
        \centering
        \includegraphics[width=\textwidth, trim = 1cm 0cm 1cm 0.5cm, clip=true]{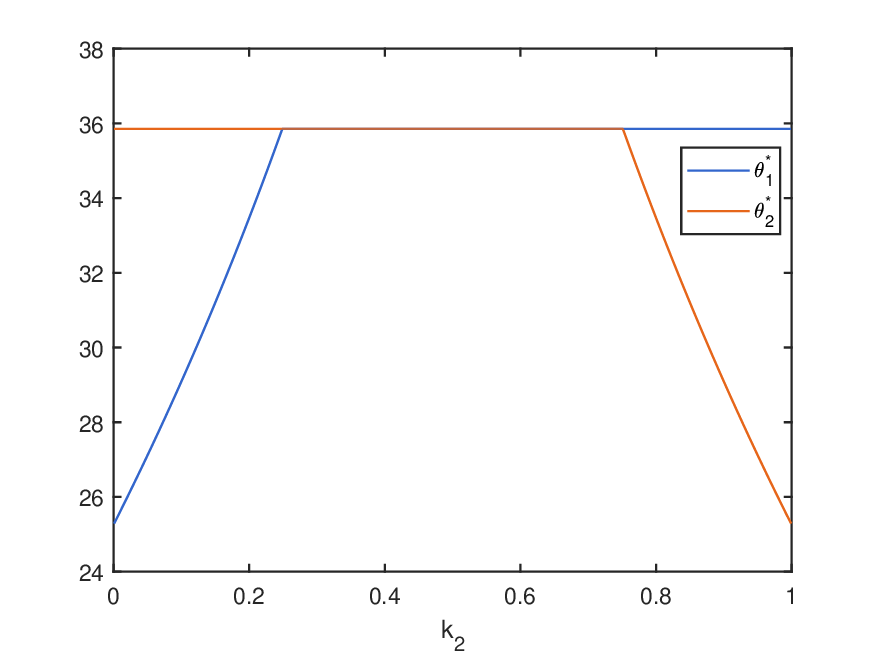}
    \end{subfigure}
    \hfill
    \begin{subfigure}[b]{0.45\textwidth}
        \centering
        \includegraphics[width=\textwidth, trim = 1cm 0cm 1cm 0.5cm, clip=true]{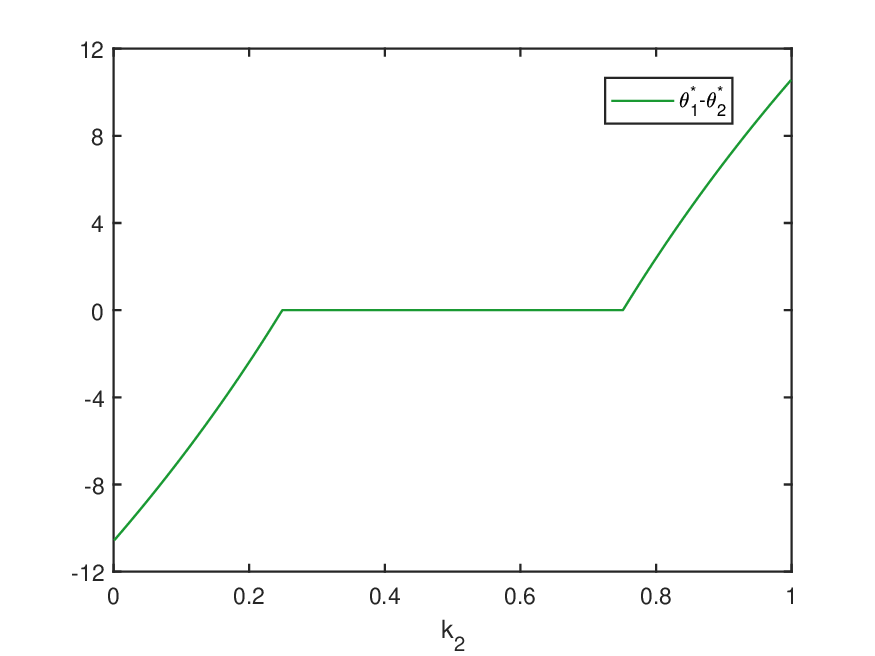}
    \end{subfigure}
    \vspace{-2ex}
    \caption{Effect of the preference coefficient \(k_2\) for Company~1  on the equilibrium premiums}
    \label{fig:both2}
\end{figure}

Next, we examine the impact of the preference coefficient $k_2$ on the equilibrium premiums and their difference. Recall from \eqref{k1k2} that $k_2$ is the probability that an insured chooses Company~1 when two companies offer the same premiums; therefore, $k_2 > 0.5$ (resp., $k_2 < 0.5$) implies a preference advantage for Company~1 (resp., Company~2). The order of the equilibrium premiums in the left panel of Figure~\ref{fig:both2} is fully consistent with the meaning of $k_2$ and confirms the theoretical result in Proposition~\ref{p4.1}. As long as the preference effect is strong enough (i.e., when $k_2$ is close to either 0 or 1), the advantaged company charges the maximum allowed premiums, but the disadvantaged company offers very low premiums in order to attract customers. It is worth pointing out that the equilibrium premiums from the two companies coincide over an \emph{interval} around $k_2 = 0.5$, not limited to $k_2 = 0.5$; this is due to the existence of an upper bound on premiums, and for the particular base parameters, the equilibrium premiums are achieved at the upper bound for a range of $k_2$ near 0.5. Indeed, if we manually lift out the upper bound $M$ to 120, the difference $\theta_1^* - \theta_2^*$ is strictly increasing (from around -30 to 30) when $k_2$ increases from 0 to 1, and the two premiums are only equal at the point $k_2 = 0.5$. (Graphs for this case are available upon request, and here we omit them to save space.)

\begin{figure}[htbp]
    \centering
    \begin{subfigure}[b]{0.45\textwidth}
        \centering
        \includegraphics[width=\textwidth, trim = 1cm 0cm 1cm 0.5cm, clip=true]{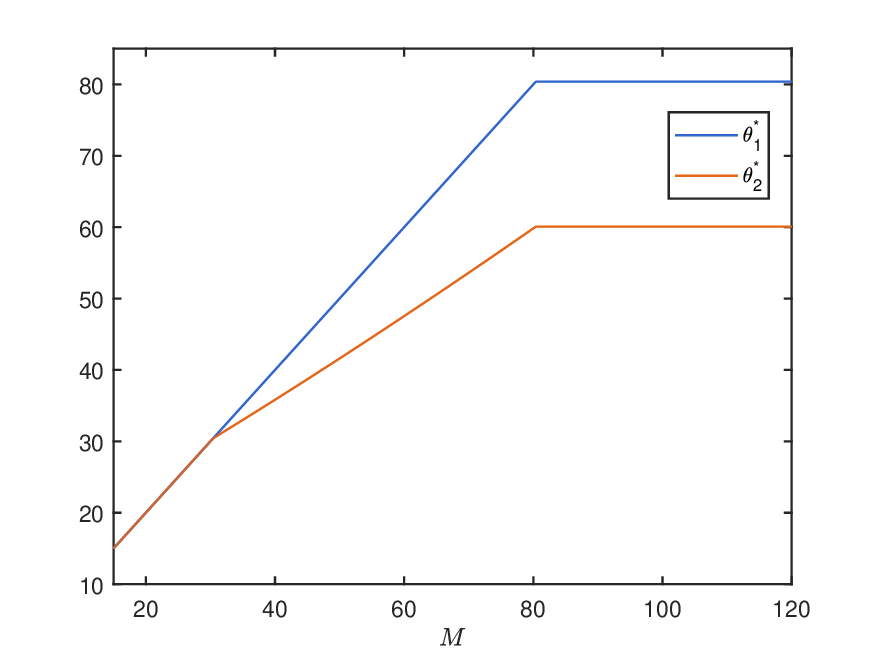}
    \end{subfigure}
    \hfill
    \begin{subfigure}[b]{0.45\textwidth}
        \centering
        \includegraphics[width=\textwidth, trim = 1cm 0cm 1cm 0.5cm, clip=true]{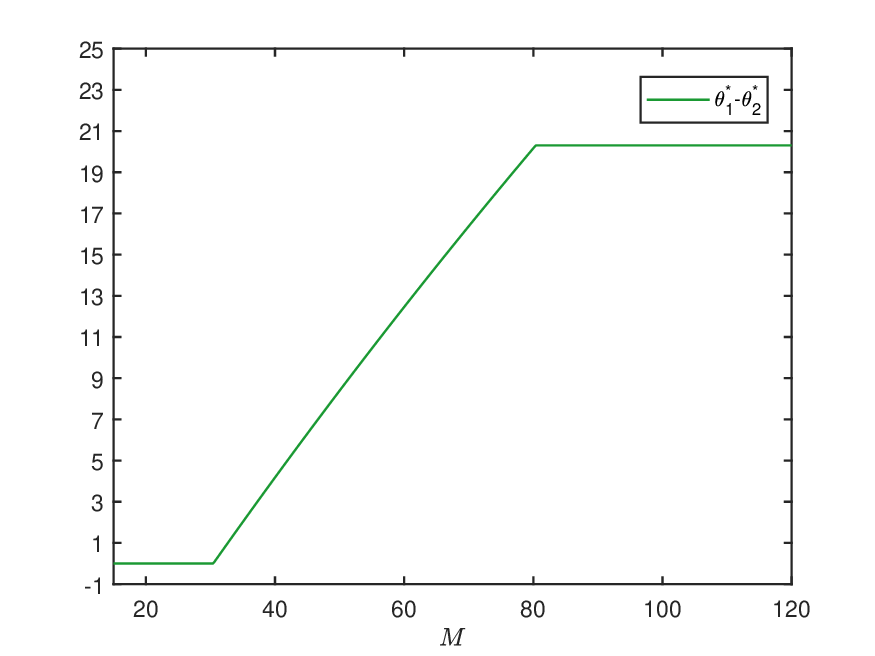}
    \end{subfigure}
\vspace{-2ex}
    \caption{Effect of the upper bound on the premium  \(M\) on the equilibrium premiums}
    \label{fig:both3}
\end{figure}

Finally, Figure~\ref{fig:both3} demonstrates the impact of the upper bound on the premium  \(M\)  on the equilibrium premiums and their difference.  Figure~\ref{fig:both3} (left panel) confirms an overall increasing relationship between the equilibrium premiums and the upper bound on the premium. In fine detail, the graph illustrates that the equilibrium premiums of the two insurance companies initially increase synchronously as the upper bound on the premium  \(M\) rises. Subsequently,  the increase in $\theta_1^*$ outpaces that of $\theta_2^*$,  and both eventually converge to stable levels. The right panel shows that when $M$ is small (specifically,  below 30.3450),  the premium difference $\theta_1^* - \theta_2^*$ remains zero; once this threshold is exceeded,  the difference increases monotonically and ultimately stabilizes at a certain level.

\section{Conclusion}\label{conl}

This paper considers a discrete-time BMS insurance model with two competing insurance companies and a continuum of insureds, aiming to understand the strategic decisions made by insureds regarding loss reporting and by insurance companies regarding BMS premiums. For given premiums, we solve the insureds' loss reporting problem and obtain an optimal barrier strategy that minimizes the expected total expenses. Next, knowing how insureds report losses, the insurance companies set their premiums to maximize their expected profit over a unit period, and their competition is settled by a non-cooperative Nash game. We find sufficient conditions that guarantee the existence of a Nash equilibrium for the 2-class BMS setup. These theoretical results provide economic understandings of the strategic underreporting and inter-company competition in a competitive insurance market. Numerical analyses further illustrate the economic implications of the model. When the insureds' price sensitivity increases, the equilibrium premiums of the two insurers converge,  indicating intensified competition and a reduction in market power. Conversely,  stronger brand preference induces persistent price asymmetry even when both companies face identical cost structures. 

%An increase in the maximum allowed premium leads to a rise in the equilibrium premiums until they stabilize, indicating that a higher upper bound on the premium allows insurers to better reflect risk levels until the market reaches its natural equilibrium. 
%Together,  these results highlight that the interplay between behavior and market structure is crucial for understanding real-world premium dynamics.

From a practical perspective,  the findings provide a theoretical foundation for actuarial pricing under competitive and behaviorally complex environments. Future research directions include extending the model to asymmetric information settings or heterogeneous loss distributions.
\vskip 10pt
\noindent
\paragraph{Acknowledgements.} Zongxia Liang is supported by the National Natural Science Foundation of China (no.12271290). The authors thank the members of the group of Mathematical Finance and Actuarial Science at the Department of Mathematical Sciences, Tsinghua University, for their helpful feedback and conversations. The authors used ChatGPT (GPT-5, OpenAI, used in December 2025) to assist in improving the grammar and writing style of this manuscript. 
\vskip 10pt
\noindent
\bibliographystyle{apalike}
\bibliography{reference} 

\appendix

\section{Proofs}
\label{sec:proofs}

\subsection{Proof of Theorem \ref{thm:b_op}}
 \begin{proof}
	Given $\mathbf{c}$,  denote the ``one-step'' cost by 
	$ r^{\mathbf{c}}_{x}(b) := \delta\Eb_{x} \left[c(x) + L_1 \mathbf{1}_{\{L_1 \leq b(x)\}} \right] $,  in which the subscript indicates the condition on the initial state $X_1 = x= (n,  i) \in \Xc$. For every $x \in \Xc$, the value function $V^*$ satisfies the Bellman equation 
	\begin{align}
		\label{eq:V_HJB}
		V^*(x) = \inf_{b \in \Bc} \,  \left\{ r^{\bc}_{x}(b) + \delta \sum_{X_2} \big( \Pb_{x, b, \bc}(X_2) \,  V^*(X_2)  \big)  \right\} := \,  \mathcal{L}^{\bc} \, V^*(x), 
	\end{align}
	where $X_2$ can only take one of the four states,  $(\ndo,  i)$,  $(\ndo,  3-i)$,  $(\nup,  i)$,  $(\nup,  3-i)$,  and the transition probabilities $\Pb_{x, b, \bc}(X_2)$, given the premium pair $\bc$, $X_1 = x$ and a barrier strategy $b$, are given by \eqref{eq:P}. 
	
Let $\Vc := \{v: \Xc \mapsto \Rb \}$ and equip $\Vc$ with the sup norm,  $\| v \| := \sup\limits_{x \in \Xc} \,  |v(x)|$ for all $v \in \Vc$, then,  $(\Vc,  \| \cdot \|)$ is a Banach space. We first show that $\mathcal{L}^{\bc} : v \in \Vc \mapsto \mathcal{L}^{\bc} v \in \Vc$ is a contraction. To that end,  fix an arbitrary $x \in \Xc$ and consider two functions $v_1,  v_2 \in \Vc$,  assume,  without loss of generality,  that $\mathcal{L}^{\bc} v_1(x) \ge \mathcal{L}^{\bc} v_2(x)$ and let $\{b_{(m)} \}_{m=1}^{\infty}$ be a sequence of barrier strategies achieving the infimum in $\mathcal{L}^{\bc} v_2(x)$. Then we obtain 
	\begin{equation}
		\label{eq:v12}
		\begin{split}
			&\mathcal{L}^{\bc} v_1(x) -   \mathcal{L}^{\bc} v_2(x) \\
			%=&  \inf_{b\in\Bc}\left(r_{x}^{\mathbf{c}}(b)+\delta \sum_{X_2 \in \Xc}\Pb_{x, b}(X_2) v_1(X_2) \right)-\inf_{b\in\Bc}\left(r_{x}^{\mathbf{c}}(b)+\delta \sum_{X_2 \in \Xc}\Pb_{x, b}(X_2) v_2(X_2)\right) \\
			=&  \inf_{b\in\Bc} \Big(r_x^{\mathbf{c}}(b)+\delta \sum_{X_2 \in \Xc}\Pb_{x, b, \bc}(X_2) v_1(X_2) \Big)-\lim\limits_{m \to +\infty} \Big(r_{x}^{\mathbf{c}}(b_{(m)})+\delta \sum_{X_2 \in \Xc}\Pb_{x, b_{(m)}, \bc}(X_2) v_2(X_2) \Big)\\
			\leq &  \liminf_{m \rightarrow +\infty} \bigg[ \Big(r_{x}^{\mathbf{c}}(b_{(m)})+\delta \sum_{X_2 \in \Xc}\Pb_{x, b_{(m)}, \bc}(X_2) v_1(X_2) \Big)-\Big(r_{x}^{\mathbf{c}}(b_{(m)}) +\delta \sum_{X_2 \in \Xc}\Pb_{x, b_{(m)}, \bc}(X_2) v_2(X_2) \Big)\bigg]\\
			%  =& \liminf_{m \rightarrow +\infty}\delta \sum_{x_0 \in S}\Pb_{x, b_m}(x_0)(v_1(x_0)-v_2(x_0))\\
			\leq &  \delta \, \liminf_{m \rightarrow +\infty} \sum_{X_2 \in S}\Pb_{x, b_{(m)}, \bc}(X_2) \big( v_1(X_2) - v_2(X_2) \big) \le \delta \|v_1-v_2\|.
		\end{split}
	\end{equation}
	Similarly, if $\mathcal{L}^{\mathbf{c}}v_1(x) \leq \mathcal{L}^{\mathbf{c}}v_2(x)$,  then $0 \leq \mathcal{L}^{\mathbf{c}}v_1(x) - \mathcal{L}^{\mathbf{c}}v_2(x) \leq \delta \|v_1-v_2\|$. Because $\delta \in (0,1)$, $\mathcal{L}^{\mathbf{c}}$ is a contraction mapping on $(\Vc,  \| \cdot \|)$ as claimed. By Banach Fixed-Point Theorem,  it follows that there exists a unique $v^* \in \Vc$ such that $\mathcal{L}^{\mathbf{c}} v^* = v^*$, and this $v^*$ coincides with the value function $V^*$ in \eqref{eq:V_HJB}.
	
	Next, we prove the existence and uniqueness of an optimal barrier strategy \(b^*\). From \eqref{eq:V_HJB}, we obtain, for all $(n, i) = x \in \Xc$,   
	\begin{align*}
		V^*(n, i) =\inf_{b \in \Bc }\Bigl\{& \, \delta\Eb\left[c_n^i+L\mathbf{1}_{\{L \leq b_n^i\}} \right] \\
		&+\delta \, \Pb(L \leq b_n^i) \left( 1-\etat(\Delta c^{i, 3-i}_{ \ndo }) \right) V^*(\ndo, i) +\delta \,  \Pb(L \leq b_n^i) \, \etat(\Delta c^{i, 3-i}_{ \ndo })V^*(\ndo, 3-i)\\
		&+\delta \,  \Pb(L > b_n^i) \left(1-\etat(\Delta c^{i, 3-i}_{\nup}) \right) V^*(\nup, i)+\delta \,  \Pb(L > b_n^i) \, \etat(\Delta c^{i, 3-i}_{\nup})V^*(\nup, 3-i)\Bigl\}\\
        :=\inf_{b \in \Bc }&\,Q(n,i;b).
	\end{align*}
	Note that $Q(n, i;b)$ only depends on $b_n^i$ and is independent of the rest entries in the matrix $b$. Using this result and  differentiating $Q$ with respect to $b$ yield  
	\[
	\frac{\partial Q}{\partial b} = \text{diag} \left\{\delta f_L (b_n^i)\left[b_n^i- \varphi (n,  i)\right]\right\}_{n\in \Nc, i\in \Ic}\ ,
	\]
	in which $f_L$ is the density in \eqref{F_L}, and $\varphi$ is defined in \eqref{eq:varphi}. As such, we have $b^*(n, i) = \varphi (n,  i) \vee 0$, establishing the existence and uniqueness of an optimal barrier strategy simultaneously.
	
	At the last,  we prove the continuity of $b^*$ with respect to the given premium pair $\mathbf{c}$. Let $\mathbf{c},\  \ch \in \Cc^2$ be two premium pairs from the insurance companies and denote the corresponding value functions by $V^*_{\mathbf{c}}$ and $V^*_{\ch}$, respectively. Using \eqref{eq:V_HJB}, we have 
	\begin{equation}
		\label{eq:V1}
		\begin{split}
			\|V^*_{\ch}-V^*_{\bc}\| &=\|\mathcal{L}^{\ch}V^*_{\ch} - \mathcal{L}^{\mathbf{c}}V^*_{\mathbf{c}}\| \leq 
            \|\mathcal{L}^{\ch}V^*_{\ch} - \mathcal{L}^{\ch}V^*_{\mathbf{c}}\| + 
            \|\mathcal{L}^{\ch}V^*_{\mathbf{c}} - \mathcal{L}^{\mathbf{c}}V^*_{\mathbf{c}}\|
			\\
			&\leq 
            \delta\|V^*_{\ch}-V^*_{\bc}\| + 
            \|\mathcal{L}^{\ch}V^*_{\mathbf{c}} - \mathcal{L}^{\mathbf{c}}V^*_{\mathbf{c}}\|. 
            \end{split}
	\end{equation}
	Note that $\Pb_{x,b,\bc}$ is continuous with respect to $\bc$, recalling the definition of $\mathcal{L}_{\mathbf{c}}$ in \eqref{eq:V_HJB} and  following a similar argument as in \eqref{eq:v12}, we have  $$\lim_{\ch \rightarrow \bc} \|\mathcal{L}^{\ch}V^*_{\mathbf{c}} - \mathcal{L}^{\mathbf{c}}V^*_{\mathbf{c}}\|  = 0. $$ Combining the above results yields \[
	\lim\limits_{\ch \rightarrow \bc}\|V^*_{\ch}-V^*_{\bc}\| \leq (1-\delta)^{-1}\lim\limits_{\ch \rightarrow \bc}
    \|\mathcal{L}^{\ch}V^*_{\mathbf{c}} - \mathcal{L}^{\mathbf{c}}V^*_{\mathbf{c}} \| =0, 
	\]
	which impliess that $V^*$ is continuous with respect to $\mathbf{c}$. From the explicit expression of $b^*$,  we conclude that $b^*$ is also continuous with respect to $\mathbf{c}$.
\end{proof}

\subsection{Proof of Theorem~\ref{thm:bstar}}

\begin{proof}
    Given \( (\theta_1, \theta_2) \in \Theta^2\), the corresponding premium pair is denoted by $\mathbf{c}$,
 we first solve for the insureds' optimal value function \(V^*\)  defined as follows:
    \[
    V^*(X_1) = \inf_{b \in \Bc} \; \Eb \left[ \sum_{t=1}^\infty \delta^t \left( c(X_{t}) + L_{t} \cdot \mathbf{1}_{\{L_{t} \leq b(X_t)\}} \right) \Big| X_1 \in \Xc \right].
    \]
    For every $x \in \Xc$, the value function $V^*$ satisfies the Bellman equation 
    \begin{align*}
         V^*(x) & = \inf_{b \in \Bc} \,  \left\{ r^{\bc}_{x}(b) + \delta \sum_{X_2} \big( \Pb_{x, b,\bc}(X_2) \,  V^*(X_2)  \big)  \right\} := \inf_{b \in \Bc} \, \left \{ r_x^{\bc}(b) + \delta \,  \mathcal{T}_{b} V^*(x) \right \},
    \end{align*}
    where $r^{\bc}_{x}(b):= \delta \, \Eb_{x} \left[c(x) + L_1 \mathbf{1}_{\{L_1 \leq b(x)\}} \right]$. 

    By Theorem~\ref{thm:b_op} and \eqref{eq:same_b}, there exists a unique optimal reporting strategy $b^*$ satisfying 
    \begin{align*}
    V^*(x)  &= r_x^{\bc}(b^*) + (\delta \,\mathcal{T}_{b^*} V^*)(x)\\
   b^*(s_1)  &= b^*(s_2) = b^*(s_3) = b^*(s_4) := \bar{b}^*, \\
     (\mathcal{T}_{b^*}V^*)(s_1)  &= (\mathcal{T}_{b^*}V^*)(s_2) = (\mathcal{T}_{b^*}V^*)(s_3) = (\mathcal{T}_{b^*}V^*)(s_4),
    \end{align*}
in which $s_1, s_2, s_3, s_4$ are the four states defined in \eqref{eq:four_states}. 
    Substituting the above results into Equation~\eqref{eq:varphi} and writing $\eta := \eta(\theta_1 - \theta_2)$, we obtain
    \begin{align*}
        \label{N=2_b}
        \bar{b}^* &= \delta \left[ (1-\eta) V^*(s_4) + \eta \, V^*(s_3) - (1-\eta)V^*(s_2) - \eta \, V^*(s_1)
        \right] \\
        & = \delta \left[ (1 - \eta) \, r_{s_4}^{\bc}({b^*}) + \eta  \, r_{s_3}^{\bc}({b^*}) - (1-\eta) r_{s_2}^{\bc}({b^*}) - \eta \, r_{s_1}^{\bc}({b^*})\right] \\
        & = \delta (\kappa-1)\left[\eta \, \theta_1+(1-\eta) \theta_2 \right].
    \end{align*}
\end{proof}

\subsection{Proof of Theorem \ref{thm:c_op}}

Under the optimal barrier reporting strategy $b^*$, the insured’s state process constitutes a time-homogeneous Markov chain with a finite state space. As a result, the cross-sectional distribution of customer states converges to a unique stationary distribution, which can be obtained by solving \[
(I-\mathcal{T}_{b^*}^\top) \, p = 0,
\]
together with the normalization condition $\mathbf{1}^\top p = 1$. Thus the stationary distribution is given by
\begin{align*}
   p(s_1;\theta_1,\theta_2) &=  \frac{\eta \, d_1}{d_1+d_2}=\eta \, a, \quad \quad \quad \, \, \, 
   p(s_2;\theta_1,\theta_2) =\frac{\eta \, d_1(1-\eta)}{d_1+d_2}=(1-\eta)a, \\
   p(s_3;\theta_1,\theta_2) &=\frac{b_2 d_2}{d_1+d_2}=\eta \, (1-a), \quad 
   p(s_4;\theta_1,\theta_2) = \frac{(1-b_2)d_2}{d_1+d_2}=(1-\eta)(1-a),
\end{align*}
where $a = \Pb[L\le b^*], \, b^* = \delta (\kappa-1) \left[ \eta \, \theta_1 + (1-\eta) \theta_2 \right], \, \eta$ is defined in \eqref{eq:eta}. 

Having completed the preparatory steps described above, in what follows,  adopting the perspective of the insurance company,  we prove the existence of the equilibrium stated in Definition~\ref{def_c} based on Theorem~\ref{thm:bstar}.  We first examine the maximizers of $J^1$ and $J^2$.

Fix an arbitrary $i \in \Ic, \theta_i \in \Theta$, we aim to show that the function $J^i(\theta_i; \theta_{3-i})$ admits a maximizer on $\Theta$,  and that this maximizer is unique. The argument proceeds by first analyzing the local curvature of $J^i$ at any stationary point and then combining this local result with a global monotonicity argument.

To this end,  we begin by characterizing the second-order behavior of $J^1$ at a critical point. The following lemmas provide sufficient conditions under which a stationary point is indeed a point of strict local maximum.

\begin{lemma}
    \label{f1<}
Let $\theta_2 \in \Theta$,  and suppose that there exists $\theta_1 \le \theta_2$ such that \[
\frac{\partial J^1}{\partial \theta_1}(\theta_1; \theta_2) = 0.\] 
Then 
$\frac{\partial^2 J^1}{\partial \theta_1^2}(\theta_1; \theta_2) < 0$
if  the following conditions hold:
\begin{itemize}
\item[(i)] \quad 
$\displaystyle 
\sup_{\ell \in I_L}
\frac{f_L'(\ell)}{f_L(\ell)}
\in \Big[-\frac{\ee k_1}{2 \ee -1}, \frac{k_1}{2}\Big]$ with $I_L : =\left[\delta (\kappa-1) \Eb[L], \delta (\kappa-1) \frac{M}{\kappa}\right] $;
\item[(ii)] \quad 
$\displaystyle 
(1-\delta) \frac{M}{\kappa} k_1 \le A_1$,
\end{itemize}
where the $f_L(\ell)$ denotes the density function of the loss $L$ on the positive support and $A_1$ is defined in \eqref{eq:m}.
\end{lemma}
\begin{proof}
    If $\theta_1 \le \theta_2$, then\[\eta = 1-(1-k_2) \ee^{-k_1(\theta_2-\theta_1)}, \quad \frac{\partial \eta}{\partial \theta_1}=k_1(\eta-1).\]
    The first-order condition becomes
    \begin{equation}
    \label{eq:1-order<}
         k_1 (1-\eta) \{ \theta_1 [a+\kappa (1-a)] - \Eb [L\mathbf{1}_{\{L > b^*\}}]\} = \eta \left[ a + \kappa (1-a) + (\theta_1 (1-\kappa) + b^*) q_1^1 \right],
    \end{equation}
    where 
    \begin{equation}
       a = \Pb[L\le b^*],\, q_1^1 = \frac{\partial a}{\partial \theta_1} = f_L(b^*) \frac{\partial b^*}{\partial \theta_1} = f_L(b^*) \delta (\kappa-1) [\eta + k_1 (\theta_2 - \theta_1)(1 - \eta)]>0.
    \end{equation}
We now examine the second derivative:
    \begin{align*}
        \frac{\partial^2 J^1}{\partial \theta_1^2} = &
\underbrace{-k_1 (2 - \eta) \left \{ \left[a + \kappa (1-a) \right]+(\theta_1 (1-\kappa) + b^*) q_1^1 \right\}}_{I_1}\\
&+ \underbrace{ \eta \left(\delta[\eta + k_1 (\eta-1) (\theta_1-\theta_2)]-2 \right)(\kappa-1) q_1^1}_{I_2}+ \underbrace{\eta \left[\theta_1 (1-\kappa ) + b^* \right] q_1^{11}}_{I_3},
    \end{align*}
where $q_1^{11} = \frac{\partial q_1^1}{\partial \theta_1}$. To prove $ \frac{\partial^2 J^1}{\partial \theta_1^2}<0$,  define  $h_1(\theta_1;\theta_2) = -\theta_1+\delta[\eta \, \theta_1 + (1-\eta) \theta_2]$,  we distinguish two situations based on the sign of $h_1$. Before that, we have 
\begin{align}
\label{b'}
    \frac{\partial b^*}{\partial \theta_1}  & =  \delta (\kappa-1) [\eta + k_1 (\theta_2 - \theta_1)(1 - \eta)] \le \delta (\kappa-1)\left( 1+\frac{1-k_2}{\ee^2}\right).
\end{align}

\medskip
\noindent
\textbf{Case 1:} \ $h_1 \ge 0.$

In this case, we apply Condition (i) to analyze the sign of $I_1+I_3$. Using inequality~\eqref{b'} and Condition (i), we have
\begin{align}
    \eta \, q_1^{11}&-q_1^1 k_1 (2-\eta)
    \le \left(\frac{7}{4} \eta -2 \right) k_1 q_1^1 + \eta \, f_L(b^*) \frac{\partial^2 b^*}{\partial \theta_1^2} .
    \label{I1+I3}
\end{align}
If $\theta_1 \in (\theta_2-\frac{2}{k_1}, \theta_2]$,  then using inequality~\eqref{I1+I3}, we obtain
\begin{align*}
\eta \, q_1^{11}&-q_1^1 k_1 (2-\eta) \le \left(\frac{7}{4} \eta -2 \right) k_1 q_1^1 + \delta (\kappa-1) f_L(b^*) k_1 \eta (\eta-1) \left[2 + k_1 (\theta_1 - \theta_2)\right] <0.
\end{align*}

If instead  $\theta_1 \le \theta_2-\frac{2}{k_1}$,  then $\eta \in [1-\frac{1-k_2}{\ee^2}, 1)$. As such 
\begin{align}
    \eta \, q_1^{11} & - q_1^1 k_1 (2-\eta) \le \delta (\kappa-1) k_1 f_L(b^*) \eta \left[ \left(\frac{7}{4}\eta-2 \right) \left[\eta + k_1 (\theta_2-\theta_1)(1-\eta) \right]+\frac{1-k_2}{\ee^3}\right] \\
    & < \delta (\kappa-1) k_1 f_L(b^*) \eta \left(-\frac{1}{4}+\frac{1-k_2}{\ee^3} \right) < 0.
    \label{eq:h1>=0}
\end{align}
It follows that  \[
I_1 + I_3 < (\kappa-1) h_1 (\eta \, q_1^{11} - q_1^1k_1(2-\eta))<0.
\]
Similarly,  define $h_0(\theta_1;\theta_2) = \delta \left[\eta + k_1 (\eta-1) (\theta_1-\theta_2) \right]-2$. Using $\ell\ee^{-\ell} \leq \ee^{-1}$,  we obtain 
\begin{align*}
     h_0 \leq \delta \left( 1 + \frac{1-k_2}{\ee} \right)-2 < 2 \delta -2<0 , 
\end{align*}
which implies $I_2 = \eta (\kappa-1) q_1^1 h_0 <0.$
Thus $
\frac{\partial^2 J^1}{\partial \theta_1^2} = I_1 + I_2 + I_3 < 0.$

\medskip
\noindent
\textbf{Case 2:} $h_1 < 0.$ 

If $\theta_1 \in \left[ \theta_2 - \frac{2}{k_1}, \theta_2 \right]$, as \[\frac{\partial^2 b^*}{\partial \theta_1^2} = \delta k_1 (\kappa-1) (1-\eta) [k_1 (\theta_2 - \theta_1)-2] \le  0, \, h_1 \geq (\delta-1) \frac{M}{\kappa},\]  combining with inequality~\eqref{b'},  Condition (i) and Condition (ii) yields
\begin{align}
    \frac{I_2 + I_3}{\eta (\kappa-1) q_1^1} &= h_0 + h_1 \left[ \frac{f_L'(b^*)}{f_L(b^*)} \frac{\partial b^*}{\partial \theta_1} + k_1 + \frac{k_1 (\eta-2)}{\eta + k_1 (\eta-1) (\theta_1-\theta_2)} \right]  \notag\\
    & \leq h_0 + (1-\delta) \frac{M}{\kappa} k_1 \left( \frac{\ee}{2 \ee -1}\delta (\kappa-1)\left( 1+\frac{1-k_2}{\ee^2}\right) - 1 + \frac{2-\eta}{\eta + k_1 (\eta-1)(\theta_1 - \theta_2)} \right )\notag\\
    & <  h_0 + k_2 (2-\delta k_2)\frac{1}{\eta + k_1 (\eta-1) (\theta_1-\theta_2)}.
    \label{eq:h_1<0}
\end{align}
Denote $B_0 = \eta + k_1 (\eta-1) (\theta_1-\theta_2) \in \left[k_2, 1+\frac{1-k_2}{\ee^2} \right]$, then the RHS of Inequality \eqref{eq:h_1<0} can be written as \[
\delta B_0 + \frac{k_2 (2-\delta k_2)}{B_0} - 2:=g_0(B_0).
\]
We now verify that $g_0(B_0) \le 0$ on its domain. Applying the properties of the concave function,  it suffices to verify that the above expression is non-positive at the endpoints with respect to $B_0$.
\begin{align*}
    &g_0(k_2) = \delta k_2+(2-\delta k_2) - 2 = 0, \\
    &g_0\left(1+\frac{1-k_2}{\ee^2}\right) = \delta \left(1 + \frac{1-k_2}{\ee^2} \right)+\frac{k_2 (2-\delta  k_2)}{1 + \frac{1-k_2}{\ee^2}}  - 2 := \frac{\tilde{g}_0\left( k_2 \right)}{ 1+\frac{1-k_2}{\ee^2}},\\
    & \tilde{g}_0'(k_2) = 2(1-\delta k_2) + \frac{2-2 \delta}{\ee^2} + \frac{2 \delta (k_2 -1)}{\ee^4}>0,\,g_0\left(1+\frac{1-k_2}{\ee^2}\right) \le \frac{\tilde{g}_0 (1)}{1+\frac{1-k_2}{\ee^2}} = 0.
\end{align*}
Thus $g_0(B_0) \le 0$ for all admissible $B_0$,  implying $I_2 + I_3 \le 0$. 

If instead $\theta_1 <  \theta_2 - \frac{2}{k_1}$, as $\frac{\partial^2 b^*}{\partial \theta_1^2} >0$, combining with inequality~\eqref{b'} and Condition (i), we obtain
\begin{align}
    \frac{I_2 + I_3}{\eta (\kappa-1) q_1^1} & < h_0 + h_1 \frac{f_L'(b^*)}{f_L(b^*)}\frac{\partial b^*}{\partial \theta_1} \\
    & < \frac{1 }{\kappa -1} \frac{\partial b^*}{\partial \theta_1} -2 + (1-\delta) k_1 \frac{M}{\kappa}\frac{\ee}{2 \ee-1} \frac{\partial b^*}{\partial \theta_1}\\
    & < \left( \frac{1}{\kappa -1}
    +(2-\delta) \frac{\ee}{2 \ee -1} \right)\frac{\partial b^*}{\partial \theta_1}-2 <\frac{3 \ee -1}{2 \ee -1} \left(1+ \frac{1}{\ee^2}\right) -2 < 0.
\end{align}
That is, $I_2 + I_3 < 0$.  Together with Equation \eqref{eq:1-order<}:  \[
I_1 = -k_1^2 \frac{(1-\eta)(2-\eta)}{\eta} \left\{ \theta_1 [a+\kappa(1-a)] - \Eb [L \mathbf{1}_{\{L > b^*\}}] \right\} < 0, \]
we again obtain $\frac{\partial^2 J^1}{\partial \theta_1^2} < 0$. Hence,  the second derivative is strictly negative if $\frac{\partial J^1}{\partial \theta_1}=0$.
\end{proof}

\begin{lemma}
    \label{f1>}
Under the assumptions of Lemma~\ref{f1<} (Conditions (i)–(ii) in Lemma~\ref{f1<}), 
 let $\theta_2 \in \Theta$,  and suppose that  there exists $\theta_1 \ge \theta_2$ such that \[
\frac{\partial J^1}{\partial \theta_1}(\theta_1; \theta_2) = 0.\] 
Then  $
\frac{\partial^2 J^1}{\partial \theta_1^2}(\theta_1; \theta_2) < 0$
if the following condition holds:
\begin{itemize}
\item[(iii)] \quad 
$ \displaystyle 
\sup_{\ell \in I_L}  f_L(\ell) \le \mm_1$ with $I_L : =[\delta (\kappa-1) \Eb[L], \delta (\kappa-1) \frac{M}{\kappa}]$.

\end{itemize}
where $f_L(\ell)$ denotes the density function of $L$ on the positive support and $\mm_1$ is defined in \eqref{eq:m}.
\end{lemma}

\begin{proof}
    In this regime, \[
\eta = k_2 \ee^{-k_1 (\theta_1-\theta_2)}, \quad \frac{\partial \eta}{\partial \theta_1} = -k_1 \eta.
\]
The first-order condition becomes 
\begin{equation}
\label{equ:1-order>}
     k_1 \{\theta_1 [a + \kappa (1-a)]-\Eb [L \mathbf{1}_{\{L > b^*\}}]\} = a + \kappa (1-a) + [\theta_1 (1-\kappa) + b^*] q_1^1, 
\end{equation}
where \[
 a = \Pb[L \le b^*],\, q_1^1 = \frac{\partial a}{\partial \theta_1} = f_L(b^*) \frac{\partial b^*}{\partial \theta_1} =  f_L(b^*) \delta (\kappa - 1) \eta [1-k_1 (\theta_1-\theta_2)],\]
 the second derivative can be written as 
\begin{align*}
       \frac{1}{\eta} \cdot \frac{\partial^2 J^1}{\partial \theta_1^2} = &
        \underbrace{q_1^1 (\kappa-1) \left\{\left(  \frac{f_L'(b^*)}{f_L(b^*)} \frac{\partial b^*}{\partial \theta_1} - k_1\right) h_1 + \delta \, \eta [1 - k_1 (\theta_1-\theta_2)] - 2 \right\}-k_1 (\kappa-1) (1-a)}_{I_4}\\
      & \underbrace{ - k_1 + (\kappa-1) h_1 f_L(b^*) \frac{\partial^2b^*}{\partial  \theta_1^2}}_{I_5} .
    \end{align*}
where $h_1(\theta_1;\theta_2) = -\theta_1+\delta[\eta \, \theta_1 + (1-\eta) \theta_2]$. We now show separately that $I_4 < 0$ and $I_5 \le 0$ under the stated conditions.

\medskip
\noindent
\textbf{Step 1:} \ $I_4 < 0.$

We begin  by distinguishing between the following two cases,  according to the sign of  \(q_1^1 \). 

If $ \theta_1 \in \left[ \theta_2+\frac{1}{k_1}, +\infty \right)$, then \(q_1^1 \leq 0\),  combining with Condition (i),   we obtain 
\begin{align}
\label{7}
    \left( \frac{f_L'(b^*)}{f_L(b^*)} \frac{\partial b^*}{\partial \theta_1}- k_1 \right) h_1 & > - \frac{k_1}{2} h_1 = \frac{1-\delta}{2} \theta_1 k_1 + k_1 (\theta_1-\theta_2) \frac{\delta (1-\eta)}{2} \geq \frac{1}{2}-\frac{\delta k_2}{2 \ee}, 
\end{align}
in which the last inequality follows from \(\theta_1 k_1>1\) and \((1-\eta) k_1 (\theta_1-\theta_2) \geq 1 - \frac{k_2}{\ee}.\)

Denote \(y = \delta \eta [1-k_1 (\theta_1-\theta_2)]\),  then \(y \in \left[ -\frac{\delta k_2}{\ee^2}, 0 \right).\)  Combining with Inequality \eqref{7}, we obtain 
\begin{align*}
    I_4 &\leq f_L(b^*) (\kappa-1)^2 y \left(y - \frac{3}{2} - \frac{\delta k_2}{2 \ee}\right) - k_1 (\kappa-1) (1-a) \\
    &\leq f_L(b^*) (\kappa-1)^2 \frac{\delta k_2}{\ee^2} \left( \frac{\delta k_2}{\ee^2} + \frac{\delta k_2}{2 \ee} + \frac{3}{2} \right) - k_1 (\kappa-1) (1-a) \\
    &:=g_1(b^*).\\
    \frac{\partial g_1(b^*)}{\partial b^*}& = f_L(b^*) \left[ \frac{f_L'(b^*)}{f_L(b^*)} (\kappa-1)^2 \frac{\delta k_2}{\ee^2} \left( \frac{(2 + \ee) \delta k_2}{2 \ee^2}+\frac{3}{2}\right)+k_1 (\kappa-1)\right]\\
    &\geq f_L(b^*) k_1 (\kappa-1) \left[1-\frac{\ee}{2 \ee-1}(\kappa-1) \frac{\delta k_2}{\ee^2} \left( \frac{(2+\ee)\delta k_2}{2\ee^2} + \frac{3}{2} \right) \right] > 0.    
\end{align*}
Thus 
 $ g_1(b^*) < \lim \limits_{b^*\rightarrow + \infty} g_1(b^*)=0,$
i.e., \(I_4 <0.\)

If \(\theta_1 \in \left[\theta_2, \theta_2 + \frac{1}{k_1} \right) \),  then \(q_1^1 > 0\). Combining with Condition (i) and Condition (iii),  we obtain 
\begin{align*}
    \left( \frac{f_L'(b^*)}{f_L(b^*)} \frac{\partial b^*}{\partial \theta_1} - k_1 \right) h_1 & \leq \left(1+\frac{\ee \delta (\kappa-1) k_2}{2\ee-1} \right) k_1 \left[(1-\delta) \frac{M}{\kappa} + \delta (1-\eta)(\theta_1-\theta_2) \right] \\
    &\leq \frac{3 \ee-1}{2 \ee-1} \left[ \frac{k_2 (2-\delta k_2)}{2 - k_2} + \delta \left(1-\frac{k_2}{\ee} \right) \right].
\end{align*}
Noting that in this case,  \(y \in (0, \delta k_2]\), we have 

\begin{align}
    I_4 &\leq f_L(b^*) (\kappa-1)^2 y \left\{y + \frac{3\ee-1}{2 \ee-1} \left[ \frac{k_2 (2-\delta k_2)}{2-k_2} + \delta \left(1-\frac{k_2}{\ee} \right) \right] -2\right\} \\
    & \quad \, - k_1 (\kappa-1) (1-a) \\
    &\leq f_L(b^*) (\kappa-1)^2 y \left\{\delta k_2 + \frac{3\ee-1}{2 \ee-1} \left[ \frac{k_2 (2-\delta k_2)}{2-k_2} + \delta \left(1-\frac{k_2}{\ee} \right) \right] -2\right\}\\
     & \quad \, - k_1 (\kappa-1) (1-a)\label{8}\\
    &\leq f_L(b^*) (\kappa-1)^2 y \left\{\delta + \frac{3\ee-1}{2 \ee-1} \left(2-\frac{\delta}{e}  \right) -2 \right\} - k_1 (\kappa-1) (1-a)\notag\\
    &\leq f_L(b^*) (\kappa-1)^2 \delta k_2 \left\{ \delta  + \frac{3\ee-1}{2 \ee-1} \left(2-\frac{\delta}{e}  \right) -2 \right\}-k_1 (\kappa-1) (1-a) \notag\\
    &:=g_2(b^*).
    \end{align}
    \begin{align}
    \frac{\partial g_2(b^*)}{\partial b^*} &= f_L(b^*)(\kappa - 1) \left\{\frac{f_L'(b^*)}{f_L(b^*)}(\kappa-1)\delta k_2 \left[ \delta + \frac{3\ee-1}{2 \ee-1} \left(2-\frac{\delta}{e}  \right) -2 \right]\!+\!k_1 \right\}\\
    &\geq f_L(b^*) (\kappa-1) k_1 \left\{ 1-(\kappa-1) \delta k_2 \frac{\ee}{2\ee-1} \left[ \delta + \frac{3\ee-1}{2 \ee-1} \left(2-\frac{\delta}{e}  \right) -2 \right] \right\} \notag\\
    &\geq f_L(b^*) (\kappa-1 )k_1 \left[1-(\kappa-1) \delta \frac{\ee}{2 \ee - 1} \left( \frac{2 \ee^2 -4\ee +1}{(2 \ee -1) \ee} \delta + \frac{2 \ee}{2 \ee -1}\right)\right]  \notag\\
    & \ge f_L(b^*) (\kappa-1 )k_1 (2 - \kappa) > 0, \notag
\end{align}
where  the first inequality follows from that the content inside the braces in Expression \eqref{8} is monotonically increasing with respect to \(k_2\). Thus we have
$ 
 g_2(b^*) \leq \lim\limits_{b^* \rightarrow +\infty} g_2(b^*) = 0,
$ i.e.,  \(I_4 < 0.\)

\medskip
\noindent
\textbf{Step 2:} \ $I_5 \le 0.$ 

It is immediately that  if $\theta_1 \in \left( \theta_2 + \frac{2}{k_2}, +\infty \right)$,  then $\frac{\partial^2 b^*}{\partial \theta_1^2} > 0$, and consequently $I_5<0$ follows. Hence, it suffices to consider the case corresponding to \(\theta_1 \in \left[ \theta_2, \theta_2 + \frac{2}{k_1} \right].\)  Observe that using Condition (i)
\begin{align*}
    \frac{\partial \left (f_L(b^*) \frac{\partial^2 b^*}{\partial \theta_1^2} \right)}{\partial \theta_1} & = f_L(b^*) \Bigl[ \delta (\kappa-1) k_1^2 \eta [3 - k_1 (\theta_1-\theta_2)] \\
    & + \delta^2 (\kappa-1)^2 k_1 \eta^2 [k_1 (\theta_1-\theta_2) - 2][1-k_1(\theta_1-\theta_2)]\frac{f_L'(b^*)}{f_L(b^*)}\Bigl] \\
    &\ge f_L(b^*) \delta (\kappa-1) k_1^2 \eta \left[ 1-\delta (\kappa-1) \eta \right] > 0.\\
    \frac{\partial h_1}{\partial \theta_1} & = \delta \eta - 1 - \delta \eta k_1 (\theta_1-\theta_2) < 0.
\end{align*}
Let $\left[ \theta_2, \theta_2+\frac{2}{k_1} \right]$ be divided into $2 z$ equal parts, $z\in \mathbb{N}=\{1, 2, \cdots\} $. For $0 \leq m \leq 2 z - 1$ and $m \in \mathbb{N}$, using Condition(ii), it follows that on the interval $\left[ \theta_2 + \frac{m}{z k_1}, \theta_2 + \frac{m+1}{z k_1}\right]$,  
\begin{align*}
    (\kappa-1) h_1 f_L(b^*) \frac{\partial^2 b^*}{\partial \theta_1^2} & \leq (\kappa-1) \left|h_1(\theta_2 + \frac{m+1}{z k_1} ;\theta_2)\right| \cdot \left|f_L(b^*) \frac{\partial^2 b^*}{\partial \theta_1^2}(\theta_2 + \frac{m}{z k_1} ; \theta_2)\right|\\
    &=f_L(b^*) \delta (\kappa-1)^2 k_2 \ee^{-\frac{m}{z}} \left(2-\frac{m}{z} \right) \left[ k_1 (1-\delta) \theta_1 + \frac{m+1}{z} \delta (1- k_2 \ee^{-\frac{m+1}{z}}) \right]\\
    &\le f_L(b^*) \delta (\kappa-1)^2 k_2 \ee^{-\frac{m}{z}} \left(2-\frac{m}{z}\right) \left[ A_1 + \frac{m+1}{z} \delta \right]\\
    & \propto g_3 \left(\frac{m+1}{z} \right), 
 \end{align*}
 where the symbol $\propto$ denotes proportionality,  the function $ g_3(\cdot)$ is
    \[ g_3(\ell) =\ee^{-\ell + \frac{1}{z}} \left( 2+\frac{1}{z} - \ell \right) \left( A_1 + \delta \ell  \right), \quad \ell \in \left[ \frac{1}{z},2 \right], \]
  and its derivative is 
\begin{align}
    g_3'(\ell) & = \ee^{-\ell + \frac{1}{z}} \left[ \delta \ell^{2} - \left( \delta\frac{4 z + 1}{z} - A_1 \right) \ell - \left( 3+\frac{1}{z} \right) A_1 + \left(2 + \frac{1}{z} \right) \delta\right].
\end{align}
The sign of $g'_3$ is determined solely by the quadratic polynomial inside the square brackets, evaluating at endpoints gives:
\begin{align}
    g_3'(2) = \ee^{-2+\frac{1}{z}}\left[ -\left( 2+\frac{1}{z} \right) \delta - \left( 1+\frac{1}{z} \right) A_1 \right]< 0, \quad
    g_3' \left(\frac{1}{z} \right) = -3 A_1 + \left( 2-\frac{3}{z} \right) \delta.
\end{align}

If $2 \delta \le 3 A_1$, we obtain $g_3'(\ell) < 0$ for all $\ell \in \left[\frac{1}{z}, 2\right]$. Hence $g_3(\ell)$ is strictly decreasing on this interval,  and
    \[ g_3(\ell) \leq g_3 \left( \frac{1}{z} \right) = 2\left(A_1 + \frac{\delta}{z} \right)  \rightarrow 2 A_1 \, (z \rightarrow +\infty).\]

If $2 \delta > 3 A_1$, let $ z \ge \lceil \frac{3 \delta}{ 2 \delta - 3 A_1} \mathbf{1}_{\{3 A_1 < 2 \delta\}} \rceil + 1$, where $\lceil (\cdot) \rceil$ denotes the ceiling of $(\cdot)$, then $g_3'(\frac{1}{z})>0$, $g_3$ attains its maximum when $g_3'(\ell) = 0$, it follows that 
\begin{align}
    g_3(\ell) \le \ee^{-\ell + \frac{1}{z}} \left[ \left(2+\frac{1}{z} \right)\delta -A_1 -2 \delta \ell\right]  \le   \left(2 - \frac{1}{z} \right)\delta -A_1  \rightarrow (2 \delta - A_1) \, (z \rightarrow +\infty).
\end{align}

Using Condition (iii), we obtain the uniform bound
 \[
(\kappa-1) h_1 f_L(b^*) \frac{\partial^2 b^*}{\partial \theta_1^2} \leq  \sup_{\ell \in I_L} f_L(\ell) \delta (\kappa-1)^2 k_2 [(2 A_1) \vee (2 \delta -A_1)] \leq k_1.
\]
This implies $I_5\leq 0.$ 

\noindent Consequently, 
\[\frac{\partial ^2 J^1}{\partial \theta_1^2} = \eta (I_4 + I_5) < 0,\]
which is the desired result.
\end{proof}

Based on Lemmas~\ref{f1<} and~\ref{f1>}, we  establish the following theorem concerning the existence and uniqueness of the maximizer of $J^i,i \in \Ic$.

\begin{theorem}[Existence and Uniqueness of the Maximizer of $J^i, i\in \Ic$]
\label{thm:f1}
Let $i \in \Ic, \theta_{3-i} \in \Theta$, then the function $J^i$ admits a unique maximizer if the following conditions hold:
\begin{itemize}
\item[(i)] \quad 
$\displaystyle 
\sup_{\ell \in I_L}
\frac{f'_L(\ell)}{f_L(\ell)}
\in \Big[-\frac{\ee k_1}{2 \ee-1},  \frac{k_1}{2}\Big]$ with $I_L := \left[\delta   (\kappa-1)  \Eb[L], \,  \frac{\delta (\kappa-1) M}{\kappa} \right]$;

\item[(ii)] \quad 
$\displaystyle 
\sup_{\ell \in I_L} f_L(\ell) \le \mm$;

\item[(iii)] \quad 
$\displaystyle 
(1-\delta) \frac{M}{\kappa} k_1 \le A$,
\end{itemize}
where $f_L(\ell)$ denotes the density function of the loss $L$ on the positive support, $\mm = \mm_1 \mathbf{1}_{\{i =1\}} + \mm_2 \mathbf{1}_{\{i =2\}},\, A = A_1 \mathbf{1}_{\{i =1\}} + A_2 \mathbf{1}_{\{i =2\}}$,   and $\mm_1,\mm_2,A_1$ and $A_2$ are defined in~\eqref{eq:m}.
\end{theorem}
\begin{proof}
    If $i = 1$, based on  Lemma~\ref{f1<} and Lemma~\ref{f1>}, if $\frac{\partial J^1}{\partial \theta_1}$ admits a zero,  then it is unique and necessarily corresponds to a point of maximum. If $\frac{\partial J^1}{\partial \theta_1}$ has no zero,  then $J^1(\theta_1;\theta_2)$ is monotone. In this case,  the maximizer of $J^1$ over $\Theta$ exists and is unique.  In summary,  $J^1$ admits a unique maximizer over $\Theta$.

    Otherwise, if $i = 2$, we present prove the existence and uniqueness of $J^2$. The proof follows a symmetric argument to that of $i = 1$, with the roles of Company~1 and Company~2 interchanged. Corresponding adjustments are made to the function $f_L(\ell)$ and the relevant parameters (see the expression for $A$ and $\mm$), while the overall logical structure remains unchanged. For the sake of brevity, the detailed derivation is omitted.
\end{proof}

We are now ready to complete the proof of Theorem~\ref{thm:c_op}  considering all possible configurations of the first-order condition.

\begin{proof}[\bf Proof of Theorem~\ref{thm:c_op}]
 We define the mapping $\mathcal G:\Theta^2 \mapsto \Theta^2$ by
    \[\forall (\theta_1, \theta_2) \in \Theta^2, \, \ \mathcal G(\theta_1, \theta_2) = (\bar{\theta_1}, \bar{\theta_2}), \]
    where $(\bar{\theta_1}, \bar{\theta_2})$ satisfies \[
    \bar{\theta_1} = \arg\sup\limits_{\theta_1 \in \Theta} J^1(\theta_1;\theta_2), \,  \bar{\theta_2} = \arg\sup\limits_{\theta_2 \in \Theta}J^2(\theta_2;\theta_1).
    \]
    By Theorem~\ref{thm:f1},  for any $\theta_i\in\Theta$,  $i \in \Ic$, the best response $\bar \theta_j(\theta_i)$ exists, is unique, and depends continuously on $\theta_i$.
    Hence,  $\mathcal G $ is a continuous self-mapping on the compact convex set $ \Theta^2$. Thus, the existence of a fixed point follows from Brouwer’s Fixed Point Theorem.
\end{proof}

\subsection{Proof of Proposition \ref{p4.1}}

\begin{proof}
We first consider the case $k_2 \ge \tfrac{1}{2}$ and argue by contradiction.
Suppose that an equilibrium premium exists with $\theta_1 < \theta_2$.
Then the first-order conditions imply
\[
\frac{\partial J^1}{\partial \theta_1}(\theta_1;\theta_2) \le 0
\quad \text{and} \quad
\frac{\partial J^2}{\partial \theta_2}(\theta_2;\theta_1) \ge 0, 
\]
that is, 
\begin{align}
    &\eta \bigl[a + \kappa (1-a) + (\theta_1 (1-\kappa) + b^*) q_1^1 \bigr] 
    \le k_1 (1-\eta) \bigl[ \theta_1 (a + \kappa (1-a)) - \Eb[L \mathbf{1}_{\{L>b^*\}}] \bigr], 
    \label{eq:ineq1}\\
    & a + \kappa (1-a) + (\theta_2 (1-\kappa) + b^*) q_2^2 
    \ge k_1 \bigl[ \theta_2 (a + \kappa (1-a)) - \mathbb{E}[L \mathbf{1}_{\{L>b^*\}}] \bigr].
    \label{eq:ineq2}
\end{align}
Multiplying Inequality~\eqref{eq:ineq1} by $(1-\eta)^{-1}$ and subtracting
Inequality~\eqref{eq:ineq2} yield
\begin{align}
    \mathcal H :=\;&
    \left[ \frac{2 \eta - 1}{1 - \eta}-k_1 (\theta_1-\theta_2) \right][a + \kappa (1-a)]
    \nonumber\\
    &+\frac{\eta}{1-\eta} [\theta_1 (1-\kappa) + b^*] q_1^1
    - [\theta_2 (1-\kappa) + b^*] q_2^2
    \le 0,
    \label{F}
\end{align}
where
\begin{align*}
  q_1^1 &= f_L(b^*) \delta (\kappa-1) \bigl[\eta + k_1 (\theta_2-\theta_1) (1-\eta)\bigr] > 0, \\
  q_2^2 &= f_L(b^*) \delta (\kappa-1) (1-\eta) \bigl[1 + k_1(\theta_1-\theta_2) \bigr].
\end{align*}
Because $q_1^1 + q_2^2 = f_L(b^*) \delta (\kappa-1)$, it follows that $|q_1^1| > |q_2^2|$.

\medskip
\noindent
\textbf{Case 1:} \ $q_2^2 \le 0$.
Using $q_2^2 > -q_1^1$,  we obtain
\begin{align*}
    \frac{\eta}{1 -\eta} [\theta_1 (1-\kappa)& + b^*] q_1^1
    - [\theta_2 (1-\kappa) + b^*] q_2^2
    > \frac{q_1^1}{1-\eta} \frac{\delta-1}{\delta} b^*\\ &
    \ge \frac{\eta} {1-\eta} (\kappa-1)(\delta-1) f_L(b^*) b^*
    \ge -\frac{\eta}{1-\eta} [a + \kappa (1-a)].
\end{align*}
Combining this with $k_1(\theta_2-\theta_1)>1$,  we have
\[
\mathcal H > (1-\eta) \left[ k_1 (\theta_2-\theta_1) - 1 \right] \left[a+\kappa(1-a) \right] > 0, 
\]
which contradicts Inequality~\eqref{F}.

\medskip
\noindent
\textbf{Case 2:} \ $q_2^2 > 0$.
If $\theta_1 (1-\kappa) + b^* \ge 0$,  then it is immediate that $\mathcal H >0$.
It remains to consider the case $\theta_1 (1-\kappa) + b^* < 0$.
As $ b^* - (\kappa-1) \theta_1 > (\delta-1)(\kappa-1) \theta_1 \ge \frac{\delta-1}{\delta} b^*$, we have
\begin{align*}
    \frac{\eta}{1-\eta} [\theta_1(1 & -\kappa) + b^*] q_1^1 -[\theta_2 (1-\kappa) + b^*] q_2^2\\
    >&\frac{\eta}{1-\eta} [\theta_1 (1-\kappa) + b^*] q_1^1 - [\theta_1 (1-\kappa) + b^*] q_2^2\\
   >&f_L(b^*) b^* (\delta-1) (\kappa-1) \left[ \frac{2 \eta-1}{1-\eta} - k_1 (\theta_1-\theta_2) \right]\\
   \ge& - [a+\kappa (1-a)] \left[\frac{2 \eta-1}{1-\eta} - k_1 (\theta_1-\theta_2) \right] \Longleftrightarrow \mathcal H  >0, 
\end{align*}
which again implies $\mathcal H >0$,  contradicting~\eqref{F}. Consequently,  when $k_2 \ge \tfrac{1}{2}$,  Inequality~\eqref{F} cannot hold,  and thus
$\theta_1^* \ge \theta_2^*$. When $k_2 \le \tfrac{1}{2}$,  a symmetric argument between Company~1 and Company~2 yields  $\theta_1^* \le \theta_2^*$.
If $k_2=\tfrac{1}{2}$,  both  the inequalities hold,  implying $\theta_1^* = \theta_2^*$.
\end{proof}

\section{An Example}
\label{sec:example}

In the main paper, we show in Theorem~\ref{thm:c_op} that a Nash equilibrium premium strategy exists when a set of technical conditions holds. However, despite discussions in Remark~\ref{rem:cond}, it is not clear that all those conditions can hold simultaneously in a reasonable setup. To provide a (positive) answer to this question, we construct an example in which all conditions in Theorem~\ref{thm:c_op} are satisfied.

\begin{example}
	\label{exm:condition}
	Assume that the positive part of the loss follows a Gamma distribution, $L|L>0 \sim Gamma(\alpha, \lambda)$, which belongs to the exponential family of distributions and is capable of capturing the fat tails commonly observed in insurance data (see, e.g., Section 17.3.1 in \cite{frees2009regression}). As a result, the cdf of the loss random variable $L$ in \eqref{F_L} is given by 
    \begin{align*}
    	F_L(x) = p_0 + (1 - p_0 )\int_0^x    \frac{\lambda^\alpha}{\Gamma(\alpha)}\ell^{\alpha-1}\ee^{-\lambda \ell}  \dd \ell, \quad  \text{with } \Gamma(\alpha) = \int^{\infty}_{0}t^{\alpha -1}\ee^{-t}dt.
    \end{align*}
    Because the underlying losses are from non-life risks, we assume that $p_0 > 0.5$ (for example, often more than 90\% auto insurance policies do not incur losses over a unit period).

   Regarding the upper bound $M$ on the BMS premiums, we assume a conservative condition below:
   \[M \le  3\Eb[L].\]
   Empirical evidence from property and casualty insurance markets shows that typical loss ratios $LR = \frac{\text{incurred losses}}{\text{earned premiums}}$ range from $0.6$ to $0.8$, with combined ratios close to $1$ (see \cite{allstate2023} and \cite{swissre2025}). Thus, the above bound holds for most insurance lines in practice.
   \end{example} 

Applying Theorem~\ref{thm:c_op} to the setup in Example~\ref{exm:condition} yields a refined result below.

 \begin{corollary}
	\label{cor:exist}
	Let Assumption~\ref{asu:2-class} hold and further assume the specifications as in Example~\ref{exm:condition}. There exists an equilibrium premium strategy $(\theta_1^*, \theta_2^*)$ if the following conditions hold:
 \begin{itemize}
        \item [(i)] \quad
        $\displaystyle
         \frac{k_1}{\lambda} \ge \max \left\{2-\frac{1}{\ee}, 2 \left[ \frac{\alpha-1}{\alpha \delta (\kappa-1) (1-p_0)}-1 \right],\,3(\kappa-1) \right\}$;
       \item [(ii)] \quad
       $\displaystyle
       \frac{k_1}{\lambda} \le \frac{\kappa}{3 \alpha (1-p_0) (1-\delta)}
    (A_1 \wedge A_2),$
    \end{itemize}
    \label{c4.1}
    where $A_1$ and $A_2$ are defined in \eqref{eq:m}. 
\end{corollary} 

 \begin{proof}
 First, a calculation confirms that Condition (i) of Corollary~\ref{c4.1} implies Condition (i) of Theorem~\ref{thm:c_op}.
    
    We now verify Condition (ii) of Theorem~\ref{thm:c_op}. 
    If $\alpha \geq 1$, using the bound $\Gamma(\alpha) >  \sqrt{2 \pi} \alpha^{\alpha-1} \ee^{-\alpha}$, we obtain \[
     f_L(\ell) \leq f_L \left(\frac{\alpha-1}{\lambda} \right) < \frac{\lambda \ee}{\sqrt{2 \pi}} \left( 1-\frac{1}{\alpha} \right)^{\alpha-1} (1-p_0). \]
    Since  $\left( 1-\frac{1}{\alpha} \right)^{\alpha-1} < 1$ and $ 1-p_0 <0.5$, it follows that 
    \[
    \sup_{\ell \in I_L} f_L(\ell) < \frac{\lambda e}{2 \sqrt{2 \pi}} < \frac{k_1}{2} < \mm_1 \wedge \mm_2, 
    \]
    where $\mm_1$ and $\mm_2$ are defined by \eqref{eq:m}. That is,  Condition (ii) in Theorem~\ref{thm:c_op}  holds.

If \(\alpha \in (0, 1)\), applying  the Kershaw's inequality, we have
    \begin{align}
    \label{Ker inequality}
    \Gamma(\alpha) > \left(\frac{2}{3}\right)^{1-\alpha} \frac{1}{\alpha},
     \end{align}
     which, together with Condition (iii) of Corollary~\ref{c4.1}, implies
\begin{align}
    \sup_{\ell \in I_L} f_{L} (\ell)   = f_{L} (\delta (\kappa-1) \Eb[L]) & < 
    \lambda \left( \frac{2 \alpha}{1+\alpha} \right)^{\alpha} \frac{1}{\delta (\kappa-1)} < \frac{k_1}{3 \delta (\kappa-1)^2} < \mm_1 \wedge \mm_2. 
\end{align}
Condition (ii) of Theorem~\ref{thm:c_op} is again satisfied when $\alpha \in (0,1)$.

    Finally,  we examine Condition (iii) of Theorem~\ref{thm:c_op}. Condition (ii) of Corollary~\ref{c4.1} implies
    \begin{align}
        M \le 3 (1-p_0)\frac{\alpha}{\lambda} \le \frac{\kappa}{(1-\delta) k_1} (A_1 \wedge A_2).
    \end{align}
As such, Condition (iii) of Theorem~\ref{thm:c_op} is satisfied. Thus, because all requirements of Theorem~\ref{thm:c_op} hold, we apply it to conclude that the equilibrium per Definition~\ref{def_c} exists.
\end{proof}

We discuss the conditions in Corollary~\ref{cor:exist} as follows. Note that if the parameters satisfy
    \[1-p_0 \ge 0.095,\,3 (1-\delta) < k_2 < 3 \delta -2 , \, 1.5 > \kappa >1.3,\, \alpha \le 1.05< \frac{2 \ee}{2 \ee-(4 \ee-1) \delta (\kappa-1) (1-p_0)}, \]
   then the three conditions in Corollary~\ref{c4.1} hold if the following conditions are satisfied: 
    \begin{align}
        2-\frac{1}{\ee} < \frac{k_1}{\lambda} & < 2 \le \frac{1}{1-p_0} < \frac{\kappa}{3 \alpha (1-p_0) (1-\delta)}
    \min \left\{ \frac{(1-k_2) (2-\delta (1- k_2))}{1+k_2}, \frac{k_2 (2-\delta k_2)}{2-k_2} \right\}.
    \end{align} 

We mention that the requirement on the loss distribution is practically plausible. For example, in the study of \cite{kmetic1993parametric} on parametric modeling of medical insurance claims, nearly half of the fitted Gamma distributions have shape parameters \(\alpha\) below 1.05, while the rate parameter $\lambda$ falls within the interval $(0.001,  0.01)$.
We close this section with an important remark: for the setup in Example~\ref{exm:condition}, the region of the parameters under which an equilibrium $(\theta_1^*, \theta_2^*)$ exists as in Corollary~\ref{cor:exist} is non-empty. 
\end{document}